\documentclass[11pt]{article}
\usepackage[journal=jquant,stage=submission,logging=stdout]{jsty3-author}

\usepackage{amsfonts}
\usepackage{amsmath,amssymb}
\usepackage{multirow}
\usepackage{mathtools}
\usepackage{empheq}
\usepackage{tensor}
\usepackage{cancel}
\usepackage{enumitem}
\usepackage{url}
\usepackage{simplewick}
\usepackage{placeins}

\usepackage{bbm}
\usepackage{enumitem}
\usepackage{slashed}
\usepackage{xcolor}
\usepackage{quiver}

\definecolor{THc}{rgb}{0.9,0.3,0.2}

\def\logb{\mathrm{log}_2}

\def\HH{\mathcal{H}}

\def\ZZ{\mathbb{Z}}

\usepackage{bm}
\usepackage[colorlinks=true,citecolor=blue,linkcolor=blue,urlcolor=blue]{hyperref}
\usepackage{dsfont}
\usepackage{changepage}
\usepackage{array}
\usepackage{soul}
\usepackage[capitalise]{cleveref}
\crefname{section}{Sec.}{Secs.}
\Crefname{section}{Sec.}{Secs.}
\usepackage{xifthen}
\usepackage{xargs}
\usepackage{hhline}
\usepackage{pifont}

\usepackage{amsthm}
\theoremstyle{definition}

\theoremstyle{plain}

\newtheorem{thm}{Theorem}

\newtheorem{lem}{Lemma}

\newcommand{\ba}{\begin{align}}
\newcommand{\ea}{\end{align}}
\newcommand{\be}{\begin{equation}}
\newcommand{\ee}{\end{equation}}

\DeclareMathAlphabet{\mymathbb}{U}{BOONDOX-ds}{m}{n}

\usepackage{physics}
\usepackage{braket}
\usepackage[normalem]{ulem}

\title{Long-range nonstabilizerness of topologically encoded states from mutual information}

\author{a}{David Aram Korbany}{}{}
\author{b,c,d}{Tyler D. Ellison}{}{}
\author{e}{David T. Stephen}{}{}
\author*{a}{Lorenzo Piroli}{}{}

\affiliation{a}{Dipartimento di Fisica e Astronomia, Universit\`a di Bologna and INFN, Sezione di Bologna, via Irnerio 46, I-40126 Bologna, Italy}
\affiliation{b}{Perimeter Institute for Theoretical Physics, Waterloo, Ontario N2L 2Y5, Canada}
\affiliation{c}{Department of Physics and Astronomy, Purdue University, West Lafayette, IN, 47907}
\affiliation{d}{Purdue Quantum Science and Engineering Institute, Purdue University, West Lafayette, IN, 47907}
\affiliation{e}{Quantinuum, 303 S. Technology Ct., Broomfield, Colorado 80021, USA}

\abstract{We study long-range nonstabilizerness (LRN), namely the obstruction to remove nonstabilizerness with shallow-depth local quantum circuits. In one-dimensional settings, the mutual information between disconnected spatial regions has proven to be a powerful tool to diagnose LRN. In this work, we focus on encoded states of two-dimensional topologically-ordered systems, and explore the ability of the mutual information to serve as a diagnostic of LRN. Focusing on the concrete setting of lattice models defined on a torus, we show that information about LRN can be gained from the analysis of the mutual information between non-overlapping regions containing non-contractible loops, and of the change of such mutual information under modular real-space transformations. We exemplify this idea in the toric code and the non-abelian string-net model with doubled Fibonacci topological order. In the former case, we show that the mutual information provides a full classification, certifying LRN for all encoded non-stabilizer states. In the latter case, instead, our approach does not lead to a full classification, as it detects LRN for all states except from a finite subset with special transformation properties under the modular group. Finally, we discuss how our results on LRN constrain the logical gates that can be implemented fault-tolerantly on the torus.}

\begin{document}
\nolinenumbers

\section{Introduction}

Stabilizer states have taken a prominent role in many-body physics, capturing nontrivial patterns of entanglement, despite having a simple algebraic description. In particular, stabilizer states provide explicit wave functions for topological phases of matter~\cite{kitaev2003fault, Landahl2002memory, Bombin2006colorcode, Haah2021classification, ellison2022pauli}, symmetry-protected phases~\cite{ellison2022pauli, Raussendorf2006oneway, Yoshida2016higherform, Devakul2019fractalspt, stephen2020sspt3d}, and other exotic orders~\cite{Chamon2005Chamonmodel, Haah2011Haahcode, Vijay2016xcube}. Moreover, given the development of quantum hardware for quantum error correction, they are especially natural to simulate on near-term devices~\cite{Google2023surfacecode, Iqbal2024toric, Iqbal2025qutrit, Zhu2023Nishimoriscat, Jiang2026clusters}.

In parallel, there has been growing interest in the notion of nonstabilizerness~\cite{white2021conformal, ellison2021symmetry,sarkar2020characterization,sewell2022mana, oliviero2022magic,liu2022many,haug2022quantifying, tarabunga2023many,tarabunga2024critical,falcao2024nonstabilizerness,tarabunga2024nonstabilizerness2}, which quantifies the degree to which a given state differs from a stabilizer state. In many-body settings, nonstabilizerness has been studied both for systems in~\cite{white2021conformal,ellison2021symmetry,sarkar2020characterization,liu2022many,haug2022quantifying,tarabunga2023many,tarabunga2024critical,falcao2024nonstabilizerness} and out of equilibrium~\cite{leone2022stabilizer,leone2021quantum,haferkamp2022random,haug2024probing,lopez2024exact,turkeshi2024magic,dowling2024magic}. There has also been a concerted effort to understand its interplay with entanglement~\cite{tirrito2024quantifying,gu2024magic,fux2023entanglement,frau2024nonstabilizerness,bejan2024dynamical,tarabunga2024magictransition,iannotti2025entanglement,hallam2026spectral} and quantum chaos~\cite{lami2024quantum,turkeshi2023measuring,leone2021quantum,leone2023nonstabilizerness,leone2023phase,garcia2023resource,turkeshi2025pauli,bera2025non}.

Throughout the study of nonstabilizerness, it is important to keep in mind that, contrary to several properties of interest in many-body systems, it is a basis dependent quantity. The presence of spatial locality, however, allows for a more refined notion of nonstabilizerness: one can ask which aspects of nonstabilizerness persist under local changes of basis. This motivates the introduction of long-range nonstabilizerness~\cite{white2021conformal,ellison2021symmetry,sarkar2020characterization,korbany2025long,wei2025long,parham2025quantum} (LRN), which is an obstruction to removing nonstabilizerness (at least approximately) through a local change of basis, as modeled by a shallow-depth quantum circuit~\cite{chen2010local}.\footnote{See also Refs.~\cite{leone2022stabilizer,sarkar2020characterization,tarabunga2023many,frau2024stabilizer, lopez2024exact,tarabunga2024magic,nehra2025topological,catalano2025resource,collura2026nonlocal,iannotti2026non} for different notions of non-local nonstabilizerness.}

Operationally, LRN probes an obstruction to efficiently simulating a state within the stabilizer formalism. This is to say that, if a state does not have LRN, then one can simulate it efficiently by first preparing a stabilizer state then applying a shallow-depth circuit. This is particularly relevant in the era of noisy intermediate-scale quantum devices~\cite{preskill2018quantum}, wherein it is often easier to prepare a stabilizer state on hardware. This is reminiscent of an intrinsic sign problem~\cite{smith2020intrinsic, Golan2020signproblem, ellison2021symmetry}, which informs us about a fundamental obstruction to simulating states using Monte-Carlo sampling techniques.\footnote{However, see Ref.~\cite{shackleton2026twistedquantumdoublessign}.}

At this early stage, our understanding of LRN is limited. So far, rigorous results have been obtained for one-dimensional systems using matrix-product state (MPS) techniques~\cite{korbany2025long} and in higher dimensions for specific topologically ordered states~\cite{wei2025long}. In general, however, a full characterization and theory of LRN in many-body systems is still lacking.

To further solidify our understanding of LRN beyond one-dimensional systems, it is natural to consider topological phases of matter in two spatial dimensions. These are conventionally defined by equivalence classes of states related to one another by shallow-depth circuits \cite{zeng2015quantum,kim2024classifying}. This would suggest that LRN is a property of the topolgical phase, and LRN would indicate that there is no stabilizer state representative of the phase. It is widely believed, for example, that non-abelian topological orders are incompatible with stabilizer states~\cite{potter2016symmetry, wei2025long}, thus necessitating LRN. 

This definition of topological order, however, requires some care when the system is placed on a manifold with nontrivial topology, such as a torus. The topological order may then host degenerate ground states, that cannot be mapped to one another with a shallow circuit. Thus, on a torus, LRN need not be a property of the phase. For example, while the $\ZZ_2$ toric code (TC) is a stabilizer code, it is known that there exist encoded states, i.e., ground states, featuring LRN~\cite{wei2025long}.

For one-dimensional systems, the mutual information between disconnected regions has proven to be a simple and powerful tool to diagnose LRN~\cite{korbany2025long,parham2025quantum}. In particular, Ref.~\cite{korbany2025long} developed an elementary approach to certify LRN of MPS, based on the observation that the mutual information between separated regions in certain MPS is invariant under shallow quantum circuits. Since the mutual information of stabilizer states is an integer~\cite{hamma2005bipartite,hamma2005ground}, non-integer values signal the presence of LRN~\cite{korbany2025long}. 

In this work, we explore the power of this simple approach in two-dimensional topologically ordered systems. We consider, in particular, ground states of string-net models~\cite{levin2005string,Lin2021Generalized}, but emphasize that our results apply more broadly, given the relation between string-net states and ground states of gapped two-dimensional phases proven in Ref.~\cite{kim2024classifying}. We note that this approach to studying LRN in topological order is akin to the first proof of long-range entanglement in topological orders, which also focused on systems on a torus~\cite{bravyi2006lieb}.

Our diagnostic comes from considering the mutual information between two well-separated regions each containing non-contractible loops around the torus. We then observe that the two regions can be transformed by modular transformations---mappings of the torus to itself, described in more detail in the text---and the mutual information between the transformed regions can be obtained by applying modular $S$ and $T$ matrices to the ground state~\cite{smith2020intrinsic}. Here, the modular $S$ and $T$ matrices are determined by the data of the underlying anyon theory. Thus, this allows us to directly relate the anyon data to the mutual information on transformed spaces. Combined with the fact that the mutual information of stabilizer states is an integer, this gives us a means of detecting and classifying the LRN.

We demonstrate the above programme for the $\ZZ_2$ TC and the string-net model with doubled Fibonacci topological order. In the former case, our approach provides a full classification of LRN, showing that all encoded non-stabilizer states display LRN. In the latter case, instead, our approach does not lead to a full classification, since it unable to detect LRN for a finite subset of states with special transformation properties under the modular group. 

We emphasize that some of the findings in this work were obtained previously using a different approach~\cite{wei2025long}.
In particular, Ref.~\cite{wei2025long} established a full classification of LRN of the TC ground states by leveraging the Bravyi-K\"onig theorem~\cite{bravyi2013classification}. This applies only to Pauli stabilizer codes, so to make statements about more general topological orders, Ref.~\cite{wei2025long} proved a condition for LRN dependent on the dimension of the ground-state subspace. This applies to, for example, the doubled Fibonacci topological order on manifolds of genus $g\geq 2$. Notably, this does not establish LRN on a torus, which is the focus of this work.

Our main contribution lies in the realization that, in some cases, LRN can be diagnosed through elementary techniques, based exclusively on the information-theoretic properties of the wave functions. We envision that the ideas underlying our work could be useful in more general situations beyond topological orders on a torus, including systems with boundaries, topological defects, and in higher dimensions.

The rest of this work is organized as follows. We begin in Sec.~\ref{sec:preliminaries} with an introduction to stabilizer states and an explicit definition of LRN. Our approach to detect and characterize LRN is then developed in Sec.~\ref{sec:tc_1} for the special case of the $\ZZ_2$ TC, which allows us to introduce the main ideas in a concrete, simplified setting. In Sec.~\ref{sec:non_abelian}, we extend our approach to non-abelian topological orders, focusing on the doubled Fibonacci phase. 
In Sec.~\ref{sec:open_ends} we discuss the typicality and limitations of the examples, and then briefly show applications of our findings to the problem of classifying fault-tolerant gates. Finally, our conclusions are presented in Sec.~\ref{sec:outlook}.

\section{Preliminaries}
\label{sec:preliminaries}

To get started, we review stabilizer states. We put a particular emphasis on the integer quantization of their entanglement entropy, which is central to the arguments in this work. We then give a more rigorous definition of LRN.

\subsection{Stabilizer states}

We consider lattice models on $N$ qubits, with Hilbert space $\mathcal{H}_N=\otimes_{j=1}^N\mathcal{H}_j$, where $\mathcal{H}_j\simeq \mathbb{C}^2$. We will assume that the qubits are arranged over a regular lattice of dimension $D=2$. We  denote by $\mathcal{P}_N$ the set of all $N$-qubit Pauli strings including a global phase $e^{i\phi} \in \{\pm 1,\pm i\}$. A Pauli string is a product of $N$ Pauli matrices, denoted by $X_j$, $Y_j$, $Z_j$, and $\mathbbm{1}_j$ (the identity). We will also denote the local computational basis by $\{\ket{0},\ket{1}\}$. 

One can then define stabilizer states as follows. First, for a state $\ket{\psi}$, we introduce the stabilizer group 
\begin{align}
\mathrm{Stab}(\ket{\psi}) &= \{P\in \mathcal{P}_N:P\ket{\psi} = \ket{\psi}\}\,.
\end{align}
Note that $\mathrm{Stab}(\ket{\psi})$ is an abelian sub-group of $\mathcal{P}_N$, and that its elements can only have global phases $\pm 1$. Then, we say that a state $\ket{S}$ is a stabilizer state if $|\mathrm{Stab}(\ket{S})| = 2^N$. In this case, we can write~\cite{klappenecker2002beyond}
\be\label{eq:stab_state}
\ket{S} \bra{S} = \frac{1}{2^N}\sum_{g \in \mathrm{Stab}(\ket{S})} g\,.
\ee
Therefore, the reduced density matrix over a region $R$, $\rho_R = \tr_{\bar{R}} \ket{S} \bra{S}$, is of the form
\be\label{eq:reduced_density_matrix}
\rho_R=\frac{1}{2^{|R|} }\sum_{h \in H_R} h\,,
\ee
where $H_R$ is the abelian subgroup of $\mathcal{P}_N$ containing the Pauli strings for which the trace in Eq.~\eqref{eq:stab_state} is nonzero, namely those of the form $g = h_R \otimes \mathbbm{1}_{\bar{R}}$.

From the expressions above, one can compute the bipartite entanglement entropy of a pure stabilizer state. Indeed, the von Neumann entropy $S(\rho) = -\tr[\rho \log_2(\rho)]$ of the reduced density matrix in Eq.~\eqref{eq:reduced_density_matrix} reads~\cite{hamma2005bipartite}
\begin{equation}\label{eq:von_neumann}
    S(\rho_R)=|R|-\log_{2}|H_R|\,,
\end{equation}
where we denoted by $|R|$ the number of qubits in $R$ and by $|H_R|$ the size of the subgroup $H_R$. In addition, $|H_R|=2^{\ell}$ for some integer $\ell$, so that $S(\rho_R)$ is an integer. As a consequence, the bipartite entanglement entropy of a pure stabilizer state is quantized. 

Furthermore, the mutual information between two disjoint regions $A$ and $B$ is an integral combination of entanglement entropies: $I_{A,B}(|\psi\rangle)=S(\rho_A)+S(\rho_B)-S(\rho_{AB})$. Therefore, the mutual information of a stabilizer state is also quantized as an integer. This observation will be at the basis of our characterization of LRN.

\subsection{Long-range nonstabilizerness}

Next, we recall the definition of LRN from Ref.~\cite{korbany2025long}. To this end, we consider sequences of states $\{\ket{\psi_N}\in \mathcal{H}_N\}_{N \in \mathbb{N}}$ for increasing $N$, and define the family of local quantum circuits (QCs) as the unitary operators $ Q_\ell =  V_{\ell} \ldots V_2 V_1$, where each ``layer'' $V_{n}$ contains arbitrary unitary gates acting on disjoint sets of two neighboring qubits. The integer $\ell$ is the depth of the circuit. Note that, similar to Ref.~\cite{wei2025long} but contrary to Ref.~\cite{parham2025quantum}, we require geometrically local gates. 

Following Ref.~\cite{korbany2025long}, we then say that a family of states $\{\ket{\psi_N}\}_{N \in \mathbb{N}}$ has short-range nonstabilizerness if for all $\varepsilon_0>0$ and all $\alpha>0$, there exists a local QC $Q_{D_N}$ of depth $D_N =O( {\rm polylog}(N))$ and a stabilizer state $\ket{S_N}$ such that for sufficiently large $N$
\be \label{eq:sm}
\Delta(Q_{D_N} \ket{\psi_N},\ket{S_N})\leq  \frac{\varepsilon_0}{ N^\alpha} = \varepsilon_N,
\ee
where $\Delta(\ket{\psi},\ket{\phi})=\sqrt{1-|\braket{\psi|\phi}|^2}$ is the trace distance. In other words, there exists a shallow-depth circuit $Q_{D_N}$ that maps $\ket{\psi_N}$ to a stabilizer state within an error $\varepsilon_N$. Otherwise, we say that $\{\ket{\psi_N}\}_N$ has LRN. Note that, compared to earlier definitions~\cite{ellison2021symmetry},  we allow for some error $\varepsilon_N$ that vanishes in the thermodynamic limit. As explained in Ref.~\cite{korbany2025long}, the error scaling is such that the states $Q_{D_N} \ket{\psi_N}$ and $\ket{S_N}$ can not be distinguished by the corresponding expectation values of extensive observables and their powers. 

\section{The toric code}
\label{sec:tc_1}

We begin this section by introducing the $\mathbb{Z}_2$ TC, and outline our strategy to detect and classify the LRN of its ground states (Sec.~\ref{sec:the_model}). Our approach relies on the so-called minimum-entropy states (MES)~\cite{zhang2012quasiparticle}, oftentimes also called minimally-entangled states. They are discussed in Sec.~\ref{sec:the_MES}, where we also review their relevant properties. We present and prove our main results in Secs.~\ref{sec:LRN_in_TC_GS} and~\ref{sec:full_classification_TC}.

\subsection{Overview of the problem}
\label{sec:the_model}

We consider a rectangular lattice and denote by $E$ and $V$ the set of edges and vertices, respectively. We place qubits on the edges, so that the total Hilbert space is $\mathcal{H}=\bigotimes_{e\in E} \HH_e$. The TC Hamiltonian~\citep{kitaev2003fault} is
\begin{align}\label{eq:toric_code_ham}
H &= -\sum_{v} A_v - \sum_p B_p\,,
\end{align}
where
\begin{align}
 A_v &=\bigotimes_{e \in v} Z_e\,,\qquad B_p =  \bigotimes_{e \in p} X_e\,.
\end{align}
Here, $X_e$ and $Z_e$ denote the Pauli matrices acting on the edge qubit $e$, while $e \in v$ ($e \in p$) denotes the edges incident on $v$ (surrounding the plaquette $p$). The sums are over all the vertices $v$ and plaquettes $p$, as depicted in Fig.~\ref{fig:tc1}. Every ground state $\ket{\psi}$ satisfies
\be A_v\ket{\psi} = B_p\ket{\psi}= \ket{\psi}.
\ee

\begin{figure}
\includegraphics[width =0.9\linewidth]{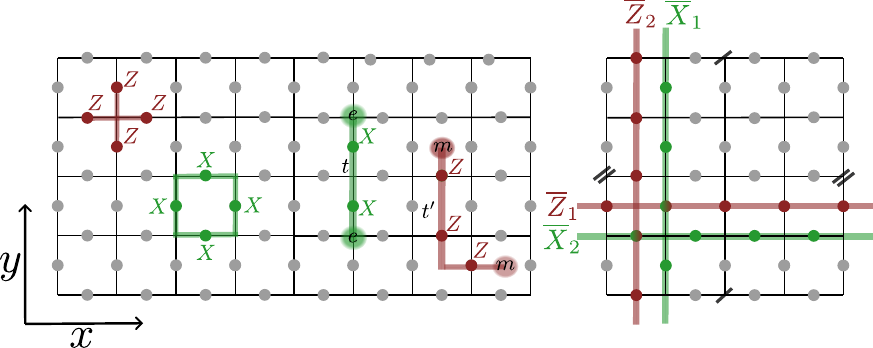}
\caption{Left: vertex, plaquette, and string operators in the TC. The string operators carry anyonic excitations at their endpoints. Right: string operators supported on non-contractible cycles on the primal and dual lattice of the torus. Here, we omit the $Z$ and $X$ operators at every edge, which are shown for the open strings on the left.}
\label{fig:tc1}
\end{figure}

We consider the model on the torus, where the dimension of the ground-state subspace is four~\cite{kitaev2003fault}. One such ground state is given by
\begin{equation}
    \ket{\psi_0}\propto \prod_{p} \left(\frac{\mathbbm{1}+B_p}{2}\right)\ket{0}^{\otimes N}\,.
\end{equation}
To define a basis for the ground-state subspace, we introduce string operators $\overline{X}_{i}$ and $\overline{Z}_i$, with $i=1,2$. They are, respectively, products of $X$ and $Z$ Pauli matrices along non-contractible loops wrapping around the torus, and are supported on qubits corresponding to a \emph{primal} and a \emph{dual} lattice, following the conventions of Fig.~\ref{fig:tc1}. A complete basis for the ground space is given by 
\begin{equation}\label{eq:stabilizer_gs}
    \ket{j_1,j_2}_{L}=\overline{X}^{j_1}_1\overline{X}^{j_2}_2\ket{\psi_0}\,,
\end{equation}
where $j_k=0,1$. We refer to the ground states in Eq.~\eqref{eq:stabilizer_gs} as the logical computational-basis states, as they represent two logical qubits in the computational basis encoded in the ground-state space. Note that they are stabilizer states and therefore do not display LRN. 

We are interested in determining the LRN of an arbitrary encoded state, i.e. ground state, which we denote by
{\begin{align}\label{eq:generic_state}
\ket{\psi} &= \alpha_1 \ket{00}_L + \alpha_2\ket{10}_L + \alpha_3 \ket{01}_L + \alpha_4 \ket{11}_L.
\end{align}

Because the ground states are locally indistinguishable, it is clear that detection of LRN necessarily requires the study of non-local correlations. In the rest of this section, we show that this task can be achieved by analyzing the mutual information 
\begin{equation}\label{eq:mutual_info}
    I_{A,B}(\ket{\psi})=S(\rho_A)+S(\rho_B)-S(\rho_{A\cup B})\,,
\end{equation}
where $A$ and $B$ are non-overlapping regions containing non-contractible loops around the torus, see Fig.~\ref{fig:ttt}.
Note that a similar mutual information, dubbed long-range mutual information, was studied in Ref.~\cite{jian2015long} as a diagnostic of topological order, see also Ref.~\cite{ritz-zwilling2024finite}.

As we state precisely and prove in Sec.~\ref{sec:LRN_in_TC_GS}, the mutual information in Eq.~\eqref{eq:mutual_info} serves a diagnostic of LRN. This statement is based on the following facts:
\begin{enumerate}
    \item For stabilizer states, the mutual information is an integer, i.e., $I_{A,B}\in \mathbb{N}$;
    \item For ground states, the mutual information is invariant under shallow quantum circuits. This fact is proven in Sec.~\ref{sec:LRN_in_TC_GS}.
\end{enumerate}
Therefore, any non-integer value of $I_{A,B}(\ket{\psi})$, for suitably chosen regions, immediately implies the LRN of the ground state $\ket{\psi}$. We note that the idea of using deviations of the mutual information from an integer value to detect LRN was first developed in Refs.~\cite{korbany2025long,parham2025quantum}.

One could wonder whether the converse statement is true, since, a priori, a state with LRN may display an integer $I_{A,B}(\ket{\psi})$. We show for the TC that this does not happen: namely, if $I_{A,B}(\ket{\psi})$ is an integer for all choices of regions $A$ and $B$ containing non-contractible loops, then $\ket{\psi}$ is a stabilizer state, and therefore does not have LRN. Putting it all together, we will obtain that the mutual information in Eq.~\eqref{eq:mutual_info} provides a faithful witness of LRN in the TC ground-state subspace.
It is useful to note that in Ref.~\cite{wei2025long}, in contrast to our approach, the LRN of encoded states in the TC was proven based on the Bravyi-K\"onig theorem~\cite{bravyi2013classification}. 

\subsection{Minimal-entropy states (MES)}
\label{sec:the_MES}

To carry out the programme outlined above, we need to introduce a convenient basis in the ground-state subspace, the MES~\cite{zhang2012quasiparticle}. To this end, we first recall the anyonic excitations of the TC. Elementary excitations, or anyons, are present if at least one of the constraints $A_v\ket{\psi} = \ket{\psi}$ or $B_p \ket{\psi} = \ket{\psi}$ is violated. When a vertex or plaquette constraint is violated, we call the excitation an electric charge $e$ or magnetic flux $m$, respectively. Excitations are created in pairs at the endpoints of string operators supported on open paths on the primal (dual) lattice as depicted in Fig.~\ref{fig:tc1}.  We further denote the trivial anyon as $I$ (no constraint violated) and the \emph{fusion} of an $e$ and $m$ particle as $f=e\times m$, i.e., there is both a vertex and plaquette violated. We can then label the anyons with an index $a\in \mathcal{A}=\{I,m,e,f\}$. 

MES are defined as the states with a well-defined anyon flux across a non-contractible loop~\cite{zhang2012quasiparticle}. To be precise, let us denote by $\gamma$ ($\gamma^\prime)$ a non-contractible loop on the primal (dual) lattice. We define the Wegner-Wilson-loop operators (WLO)
\begin{equation}
    \overline{X}(\gamma)=\bigotimes_{e\in \gamma} X_e\,,\qquad \overline{Z}(\gamma')=\bigotimes_{e\in \gamma^\prime} Z_e\,.
\end{equation}
For instance, when $\gamma=\gamma_y$ and $\gamma^\prime=\gamma'_y$ are vertical paths along the $y$ direction, then $\overline{X}(\gamma_y)$, $\overline{Z}(\gamma'_y)$ coincide with the logical Pauli operators $\overline{X}_1$, $\overline{Z}_2$, as shown in Fig.~\ref{fig:ttt}. These operators measure the flux corresponding to the anyons $e$ and $m$, respectively, with eigenvalue $1$, $(-1)$ signaling the absence (presence) of a flux for the corresponding anyon moving along the $x$-direction.

The MES are then the simultaneous eigenstates of the WLO $\overline{X}(\gamma)$ and $\overline{Z}(\gamma')$ for some non-contractible loop $\gamma$. For concreteness, consider the MES associated with the WLO $\overline{X}(\gamma_y)$ and $\overline{Z}(\gamma'_y)$. The simultaneous eigenstates are, in terms of the logical computational basis states,
\begin{equation}\label{eq:pm_j_l}
\ket{\pm,j}_L=\frac{(\ket{0,j}_L\pm \ket{1,j}_L)}{\sqrt{2}} \,. 
\end{equation}
Note that the stabilizers $A_v$, $B_p$, and $\pm \overline{X}(\gamma_y)$, $\pm \overline{Z}(\gamma'_y)$ uniquely define the MES for $\gamma_y$ and $\gamma_y'$. The $\pm$ on the WLO capture the four possible basis states of the TC ground-state subspace.

Since the four states in \eqref{eq:pm_j_l} are related to the four possible anyons $a\in \mathcal{A}$ moving along the $x$-direction, we can denote the states by $\{\ket{a;\gamma_y}\}_{a\in \mathcal{A}}$, where $a\in \mathcal{A}$ labels the anyon type. For instance, $|m;\gamma_y\rangle = \ket{+,1}_L$ is an eigenstate of $\overline{X}(\gamma_y)$ and $\overline{Z}(\gamma'_y)$, with eigenvalues $1$ and $-1$, respectively. The other cases are shown in Tab.~\ref{tab:placeholder}.

\begin{table}[t]
\centering
    \renewcommand{\arraystretch}{2}
    		\begin{tabular}{c|c} 
        Anyon label &Logical states \\
        \hline
        $\ket{I;\gamma_y}$ & 
         $\ket{+,0}_L = \frac{1}{\sqrt{2}}(\ket{0,0}_L + \ket{1,0}_L)  $  \\
 $\ket{e;\gamma_y} = \overline{Z}(\gamma_x)\ket{I;\gamma_y}$ & $\ket{-,0}_L = \frac{1}{\sqrt{2}}(\ket{0,0}_L -\ket{1,0}_L) $ \\
  $\ket{m;\gamma_y}=\overline{X}(\gamma_x)\ket{I;\gamma_y}$  & $\ket{+,1}_L= \frac{1}{\sqrt{2}}(\ket{0,1}_L +\ket{1,1}_L)  $ 
  \\
  $\ket{f;\gamma_y}=\overline{Z}(\gamma_x)\overline{X}(\gamma_x)\ket{I;\gamma_y}$ & $\ket{-,1}_L= \frac{1}{\sqrt{2}}(\ket{0,1}_L -\ket{1,1}_L)$ 

    \end{tabular}
    \caption{Correspondence between the anyon labels for the MES $\ket{a;\gamma_y}$ and the computational logical basis states~\eqref{eq:stabilizer_gs}.}

    \label{tab:placeholder}
\end{table}

\subsection{Mutual information in the MES basis}

To make statements about the LRN of the TC ground states, we compute the mutual information between two well-separated regions that wrap around non-contractible paths of the torus, as in Fig.~\ref{fig:ttt}. In this section, we provide an explicit expression for the mutual information using the MES basis. Focusing for concreteness on the path $\gamma_y$, an arbitrary ground state of the TC in Eq.~\eqref{eq:generic_state} is written in the the corresponding MES basis as
\begin{equation}\label{eq:expansion_initial_state}
    \ket{\psi} = \sum_a \psi_a(\gamma_y)\ket{a;\gamma_y}.
\end{equation}

We next consider the reduced density matrix of $\ket{\psi}$ on the two disjoints regions $A_y$ and $B_y$ on which $\overline{X}_1$ and $\overline{Z}_2$ are supported, as in the middle panel of Fig.~\ref{fig:ttt}. The reduced density matrix on $A_y$ is
\be\label{eq:rho_blocks}
\rho_{A_y} = \bigoplus_{a} |\psi_a(\gamma_y)|^2 \rho_{A_y}^a\,,
\ee
where $\rho_{A_y}^a$ is defined as
\begin{align}
    \rho_{A_y}^a = \mathrm{Tr}_{\bar{A}_y}\left( \ket{a; \gamma_y}\bra{a;\gamma_y} \right).
\end{align}
The off-diagonal blocks of $\rho_{A_y}$ vanish due to the fact that, for any Pauli string $P$ supported in $A_y$,
\be\label{eq:to_prove_1}
\braket{a;\gamma_y|P| b;\gamma_y} = 0\,, \quad \mathrm{for}\, a\neq b.
\ee
If $P$ creates excitations, then this is true, since $\ket{a; \gamma_y}$ and $\ket{b; \gamma_y}$ are ground states. Otherwise, $P$ must belong to the stabilizer group of $\ket{a; \gamma_y}$ and $\ket{b; \gamma_y}$, implying that it is diagonal in the MES basis; hence, Eq.~\eqref{eq:to_prove_1} follows. 

For our purposes, it is important to recall an additional property of MES corresponding to a loop $\gamma$. That is, their reduced density matrix factorizes over sufficiently distant regions containing non-contractible loops equivalent to $\gamma$~\cite{jian2015long,lee2013entanglement}. In the case of $\gamma_y$ above, this means
\be\label{eq:factorization}
\rho^a_{A_y B_y} = \rho^a_{A_y} \otimes \rho^a_{B_y}\,,
\ee
where $A_y$ and $B_y$ are taken sufficiently far away from one another. As a consequence,
\be\label{eq:final_decomposition_rho}
\rho_{A_y B_y} = \bigoplus_a |\psi_a(\gamma_y)|^2 \rho^a_{A_y} \otimes \rho^a_{B_y}.
\ee
Eq.~\eqref{eq:factorization} is true because
\begin{equation}\label{eq:to_prove_2}
{\rm Tr}_{A_yB_y}[P_{A_y} P_{B_y}   \rho^a_{A_y B_y} ]={\rm Tr}_{A_y}[ P_{A_y} \rho^a_{A_y} ] {\rm Tr}_{B_y}[P_{B_y}\rho^a_{B_y}]
\end{equation}
for all Pauli strings supported on $A_y$ and $B_y$. To prove Eq.~\eqref{eq:to_prove_2}, we first note that 
\begin{align}
{\rm Tr}_{A_yB_y}[P_{A_y} P_{B_y} \rho^a_{A_yB_y} ]&=\langle a;\gamma_y|P_{A_y} P_{B_y}|a; \gamma_y\rangle \,,  \\ 
{\rm Tr}_{A_y}[P_{A_y}\rho^a_{A_y} ]&=\langle a;\gamma_y|P_{A_y}|a; \gamma_y\rangle\,,\\
{\rm Tr}_{B_y}[P_{B_y}\rho^a_{B_y} ]&=\langle a;\gamma_y|P_{B_y}|a; \gamma_y\rangle\,.
\end{align}
Then, Eq.~\eqref{eq:to_prove_2} is established by treating separately the cases where $P_{A_y}$ or $P_{B_y}$ create excitations (in this case, both sides of Eq.~\eqref{eq:to_prove_2} are vanishing)  or neither of them create excitations (in this case, $P_{A_y}$ and $P_{B_y}$ act diagonally in the MES basis, and Eq.~\eqref{eq:to_prove_2} immediately follows).

So far, we have discussed the properties of MES for the vertical loop $\gamma_y$ ($\gamma_y'$). However, the same properties hold for more general non-contractible loops. In the following, we will consider three types of inequivalent loops, which we denote by $\gamma_x$, $\gamma_y$, and $\gamma_{xy}$, as depicted in Fig.~\ref{fig:ttt}. In all cases, the regions $A_{\gamma}$, $B_\gamma$ should be taken to be disjoint and sufficiently large to include the support of the WLO $\overline{X}(\gamma)$ and $\overline{Z}(\gamma')$. Note that we have already identified the WLO corresponding to $x$ and $y$, while we refer to Sec.~\ref{sec:full_classification_TC} for a detailed discussion of the case corresponding to the curve $\gamma_{xy}$.

\begin{figure}
    \centering
    \includegraphics[scale=0.8]{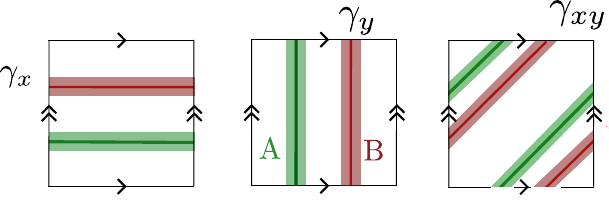}
    \caption{Three inequivalent curves $\gamma_{\alpha}$ on the torus related by modular transformations. For each curve $\gamma_\alpha$, we consider disjoint regions $A$ and $B$ containing curves that are topologically equivalent to it.}
    \label{fig:ttt}
\end{figure}

We are now prepared to compute the mutual information for regions such as those shown in Fig.~\ref{fig:ttt}. Explicitly, let us consider two disjoint regions $A_\gamma$, $B_\gamma$, containing topologically equivalent non-contractible loops. Using the decomposition in Eq.~\eqref{eq:final_decomposition_rho}, we can compute the von Neumann entropies appearing in the expressions for the mutual information, as in Eq.~\eqref{eq:mutual_info}. In particular, we have
\begin{align}
S(\rho_{A_{\gamma
} B_\gamma}) &=H(\{|\psi_{a}(\gamma)|^2\}) + \sum_a |\psi_{a}(\gamma)|^2 \left( S(\rho^a_{A_\gamma}) + S(\rho^a_{B_\gamma})\right),
\end{align}
and 
\begin{align}
S(\rho_{A_\gamma}) &= H(\{|\psi_{a}(\gamma)|^2\}) + \sum_a |\psi_{a}(\gamma)|^2  S(\rho^a_{A_\gamma})\,,\\
S(\rho_{B_\gamma})&= H(\{|\psi_{a}(\gamma)|^2\}) + \sum_a |\psi_{a}(\gamma)|^2  S(\rho^a_{B_\gamma}).
\end{align}
Here, we have have used that the supports of the density matrices $\{\rho_{A_\gamma}^a\}_a$ are orthogonal to one another, while we introduced the classical Shannon entropy
\be
H(\{p_a\}) = -\sum_a p_a \logb (p_a).
\ee
Therefore, we obtain
\begin{align}\label{eq:mut_info_Shannon} 
I_{A_\gamma,B_\gamma}(\ket{\psi}) &= H(\{|\psi_{a}(\gamma)|^2\}). 
\end{align}
That is, we recover the known fact that the mutual information coincides with the classical entropy of the coefficients $\psi_a(\gamma)$ in the MES basis~\cite{lee2013entanglement,jian2015long,shi2020fusion,wen2016edge}. 
We note that, intuitively, the mutual information is equivalent to dephasing the ground state with WLO then measuring the classical entropy of the mixture in the MES basis.

\subsection{LRN from mutual information}
\label{sec:LRN_in_TC_GS}

We now argue that a non-integral value of the mutual information in Eq.~\eqref{eq:mut_info_Shannon} is sufficient for establishing LRN in a TC ground state. We then give an explicit example of a nonstabilizer state encoded in the ground-state subspace of the TC that exhibits LRN. 

First, we establish that, for each $\gamma \in \{\gamma_x,\gamma_y,\gamma_{xy}\}$, the mutual information $ I_{A_\gamma,B_\gamma}(\ket{\psi}) $ is invariant under shallow QCs. This makes sense at an intuitive level, since the quantity $H(\{|\psi_{a}(\gamma)|^2\})$ only depends on weights associated to different ``super-selection sectors'', which cannot be changed by local operations in finite time. We formalize this claim as follows.
\begin{lem}\label{main_lemma}
Let $Q$ be a QC of depth $D_Q = O(\mathrm{polylog}(L))$ where $L={\rm min}(L_x,L_y)$ and $L_x$, $L_y$ are the linear dimensions of the system, namely $N=L_x L_y$. Let $\mathrm{width}(A_\gamma)  = \mathrm{width}(B_\gamma) \geq 2D_Q$, and suppose that the complement $C$ of $A\cup B$ is such that $\mathrm{dist}(A,B)= O(L)\gg\mathrm{width}(A_\gamma)$. Then, for any TC ground state $\ket{\psi}$, we have
\be 
I_{A_\gamma,B_\gamma}(Q\ket{\psi}) = I_{A_\gamma,B_\gamma}(\ket{\psi})\,.
\ee
\end{lem}
\noindent This lemma is a straightforward generalization of the one presented in Ref.~\cite{korbany2025long} in the context of MPS. We report a proof in Appendix~\ref{sec:lemmas_mut_info}. In the following, we will assume that $L_x$ and $L_y$ diverge with $N$, scaling as $O(\sqrt{N})$. This assumption eliminates trivial cases where the very definition of short-range nonstabilizerness may be ill-defined (this would be the case, for instance, if $L_x$ is kept fixed as $N\to\infty$).

From Lemma \ref{main_lemma}, we then arrive at our desired result:
\begin{thm}\label{thrm:lrm}
Let  $\{\ket{\psi^{(N)}}_N\} $ be a family of ground-states for increasing system size $N$. A sufficient condition for $\{\ket{\psi^{(N)}}_N\}$ to have LRN is that there exists a non-contractible loop $\gamma$ for which
\begin{equation}\label{eq:main_eq}
  \lim_{N\to \infty} I_{A_\gamma,B_\gamma}(\ket{\psi^{(N)}}_N) =   \lim_{N\to \infty}H(\{|\psi^{(N)}_{a}(\gamma)|^2\}) \notin \mathbb{N}\,.
\end{equation}
\end{thm}
\noindent This theorem follows from Lemma \ref{main_lemma}, the quantization of the mutual information for stabilizer states, and the continuity of the von Neumann entropy, as expressed by the Fannes-Audenaert inequality~\cite{audenaert2007sharp}. The proof is the same as the one presented in Ref.~\cite{korbany2025long} in the context of MPS, so we omit it here.

Let us now present an application of this theorem by considering the following encoded $T$ state, using the logical computational basis defined in Eq.~\eqref{eq:stabilizer_gs},
\be \label{eq: T state}
\ket{T}_L = \frac{1}{\sqrt{2}} \left( \ket{00}_L + \exp(i\pi/4) \ket{10}_L\right)\,.
\ee
Explicitly, there is a $T$ state encoded in the first logical qubit. We can see that this state has LRN by considering the mutual information between regions $A_y$ and $B_y$. 

The first step in calculating the mutual information between $A_y$ and $B_y$ is to rewrite the state in the MES basis for $\gamma_y$, corresponding to the eigenstates of $\overline{X}_{1}$ and $\overline{Z}_{2}$. This gives us, up to a global phase:
\begin{align}
\ket{T}_L &= \cos(\pi/8)\ket{+,0}_L - i\sin(\pi/8)\ket{-,0}_L\\
&=\cos(\pi/8)\ket{I;\gamma_y} - i\sin(\pi/8)\ket{e;\gamma_y}\,.
\end{align}
Using Eq.~\eqref{eq:mut_info_Shannon}, the mutual information between $A_y$ and $B_y$ is then
\begin{equation}
\begin{split}
H(\{|\psi_a(\gamma_y)|^2\}) = &-\cos(\pi/8)^2 \logb(\cos(\pi/8)^2) -\sin(\pi/8)^2 \logb(\sin(\pi/8)^2)\notin \mathbb{N}.
\end{split}
\end{equation}
Thus, theorem~\ref{thrm:lrm} implies that $\ket{\psi}$ has LRN. 

It is important to note that the choice of non-contractible regions matters. If we instead compute the mutual information between the regions $A_x$ and $B_x$, then the mutual information is inconclusive. To see this, note that the MES associated with the loop $\gamma_x$ are, analogously to Eq.~\eqref{eq:pm_j_l}, the four states
\begin{equation}\label{eq:j_pm_l}
\ket{j,\pm}_L=\frac{(\ket{j, 0}_L\pm \ket{j,1}_L)}{\sqrt{2}} \,. 
\end{equation}

Writing the state $\ket{T}_L$ in this basis gives
\begin{align}
    \ket{T}_L &= \frac12 \left( \ket{0,+}_L + \ket{0,-}_L + e^{i \pi/4}\ket{1,+}_L+e^{i \pi/4}\ket{1,-}_L
     \right) \nonumber \\
     &= \sum_a \psi_a(\gamma_x) \ket{a;\gamma_x}.
\end{align}
The mutual information computed using Eq.~\eqref{eq:mut_info_Shannon} is thus
\begin{align}
    H(\{|\psi_a(\gamma_x)|^2\}) = 2,
\end{align}
which would be inconclusive. 

More generally, one could consider an arbitrary encoded $T$ state of the form
\be\label{eq:general_t_state}
\ket{T}_L = \frac{1}{\sqrt{2}} \left( \ket{ij}_L + \exp(i\pi/4) \ket{i'j'}_L\right),
\ee
for two orthogonal logical computational basis states $\ket{ij}_L$ and $\ket{i'j'}_L$.
In Appendix~\ref{sec:srm_tc}, we explicitly write the change of basis from an arbitrary ground state in the computational logical basis, to the MES bases associated to the loops $\gamma_x$, $\gamma_y$, $\gamma_{xy}$ of Fig.~\ref{fig:ttt}. From inspection of these equations, it follows that one can always find \emph{exactly one} curve $\gamma\in \{\gamma_x,\gamma_y,\gamma_{xy}\}$ such that\footnote{From Eqs.~\eqref{eq:basis_trafo_y},~\eqref{eq:basis_trafo_x}, and~\eqref{eq:basis_trafo_xy}, one sees that the MES depend on all relative phases of the coefficients $\alpha_i$ of the logical computational basis \eqref{eq:generic_state}: Given $\alpha_i\neq \alpha_j$, there is one $\gamma$ such that the coefficients of the corresponding MES include $\alpha_i \pm \alpha_j$. }
\begin{align}
\ket{T}_L = \cos(\pi/8)\ket{a;\gamma} - i\sin(\pi/8)\ket{b;\gamma},
\end{align}
for some $a$ and $b$. This leads us to the same conclusion: that the state in Eq.~\eqref{eq:general_t_state} has LRN.

\subsection{Classification of LRN from modular transformations}
\label{sec:full_classification_TC}

Theorem~\ref{thrm:lrm} only gives a sufficient condition for LRN. One may wonder if there are LRN states with an integer mutual information $I_{A_\gamma,B_\gamma}$ for all choices of $\gamma$. We now show that this is not possible and that the condition in Eq.~\eqref{eq:main_eq} is necessary. 

As a preliminary consideration, we note that the values $I_{A_\gamma,B_\gamma}(\ket{\psi})$ are not independent of each other for different $\gamma$. This is because they are determined by the coefficients $\{\psi_{a}(\gamma)\}$, which are related by changes of bases. In fact, a standard argument from TQFT predicts that such changes of bases are implemented by the modular $S$ and $T$ unitary matrices that define the exchange and mutual statistics of anyons~\cite{wen1990topological,kitaev2006anyons, zhang2012quasiparticle,simon2023topological}. Specifically, the matrix element $S_{a,b}$ encodes the mutual statistics of anyons $a$ and $b$, while $T_{a,b}=\theta_a \delta_{a,b}$, where $\theta_a$ is the topological spin of type $a$ (self-statistics).

The argument relies on the properties of the modular group. In continuous space, the latter is the set of topologically distinct orientation-preserving diffeomorphisms of the torus, and is generated by two transformations traditionally denoted by $s$ (a $\pi/2$-rotation that switches the meridian and longitude of the torus) and $t$ (the Dehn twist). Given a TQFT defined on the torus, it can then be argued that these generators preserve the corresponding ground space and act as the unitary matrices $S$ and $T$~\cite{wen1990topological,kitaev2006anyons,simon2023topological}. This identification is consistent with the action of the modular transformations on non-contractible loops, and is at the basis of several works aiming at extracting topological data, including the $S$ and $T$ matrices from the ground-state wave functions~\cite{zhang2012quasiparticle,cincio2013charac,zaletal2013topological,mei2015modular}.

In our setting, the TQFT point of view is  made precise by exploiting rigorous constructions implementing the modular transformations on the lattice~\cite{zhu2020instantaneous,zhu2020quantum,smith2020intrinsic}. In the next section, we invoke these constructions when studying generic, non abelian string-net models. In the TC, however, one can follow a more elementary approach and write down explicitly the change of basis corresponding to MES of different non-contractible loops. This was done in Ref.~\cite{zhang2012quasiparticle}, but we repeat the calculations below for completeness. The knowledge of the unitary matrices implementing the changes of basis ultimately allows us to show that the only states having integer mutual information for all $\gamma$ are stabilizer states.

Recall that the states $\ket{a;\gamma_y}$ are simultaneous eigenstates of $\overline{X}_1$ and $\overline{Z}_2$, while the states
$\ket{a;\gamma_x}$ are simultaneous eigenstates of  $\overline{X}_2$ and $\overline{Z}_1$. Thus, using the explicit formulae in Eqs.~\eqref{eq:pm_j_l} and \eqref{eq:j_pm_l}, the unitary matrix which relates the two bases is
\be
S_{ab} = \braket{b;\gamma_y|a;\gamma_x} = \frac{1}{2}\begin{pmatrix}
 1 &1& 1& 1\\ 1& 1& -1& -1\\ 1& -1& 1& -1 \\1& -1& -1& 1
\end{pmatrix}\,.
\ee
This is exactly the modular $S$ matrix associated with the TC anyons~\cite{zhang2012quasiparticle}, as previously stated.

Finally, we consider loops following the regions we labeled by $xy$, in the rightmost panel of Fig.~\ref{fig:ttt}. As shown in Fig.~\ref{fig:xy_lattice}, the action of the WLOs $\overline{Z}(\gamma_{xy})$ and $\overline{X}(\gamma_{xy})$ on ground states is equal to that of $\overline{Z}_1\overline{Z}_2$ and $\overline{X}_1\overline{X}_2$, respectively, as they are related by irrelevant vertex and plaquette operators. The eigenstates $\ket{a;\gamma_{xy}}$ of these operators are Bell states in the logical computational basis, namely
\begin{align}
    \ket{\psi ^{\pm}}_L &= \frac{1}{\sqrt{2}}
    \left(\ket{00}_L \pm \ket{11}_L \right) \nonumber, \\
    \ket{\phi ^{\pm}}_L &= \frac{1}{\sqrt{2}}\left(\ket{01}_L \pm \ket{10}_L\right).
\end{align}

The change of basis now reads
\be
\braket{b;\gamma_x|a;\gamma_{xy}} =(ST)_{ab}\,,
\ee
where $T_{a,b} =\delta_{a,b}\theta_a$ is precisely the diagonal matrix of topological spins, with $\theta_a=1$ for $a\in \{I,e,m\}$ and $\theta_f=-1$. 

In summary, for a given state $\ket{\psi}$, the transformations of the four-vectors $\{\psi_a(\gamma)\}_a$ are implemented by the unitary group generated by the matrices $T$ and $S$, according to the following diagram
\be\label{eq:cd_basis_states_tc}\begin{tikzcd}
	& {|a;\gamma_y\rangle} \\
	{|a;\gamma_{xy}\rangle} && {|a;\gamma_x\rangle}
	\arrow["STS"', from=1-2, to=2-1]
	\arrow["S", from=1-2, to=2-3]
	\arrow["ST", from=2-3, to=2-1]
\end{tikzcd}
\ee
This observation allows us to complete the classification of LRN in the TC ground space. Indeed, the fact that the mutual information $I_{A_y,B_y}(\ket{\psi})$ is an integer for all choices of $\gamma$, is translated into the fact that the entropy $H(\{|\psi_a(\gamma)|^2\})$ is an integer under the action of the group of unitaries generated by the matrices $T$ and $S$. In Appendix~\ref{sec:srm_tc}, we find that the only four-vectors with this property correspond to stabilizer states, thus arriving at our second main result.
\begin{thm}\label{thrm:srm_tc}
    Let $\ket{\psi}$ be a ground state of the toric code. Then
    \be\label{eq:ness_cond_thrm}
    I_{A_\gamma,B_\gamma}(\ket{\psi}) \in \mathbb{N} \quad \mathrm{for\,all\,}\gamma
    \ee
    if and only if $\ket{\psi}$ is a stabilizer state.
\end{thm}
In the proof, we make use of the Maasen-Uffinik bound \cite{maassen1988generalized} which gives us a non-trivial lower bound on the sum \be  
I_{A_\gamma,B_\gamma}(\ket{\psi})+ I_{A_{\gamma'},B_{\gamma'}}(\ket{\psi})
\ee
for two distinct $\gamma,\gamma'$. We note that in Ref.~\cite{jian2015long}, this inequality was dubbed the topological uncertainty principle. Furthermore, we also show that the sum 
\begin{equation}
\sum_{\gamma \in \{\gamma_x,\gamma_y,\gamma_{xy}\}}
    I_{A_{\gamma},B_{\gamma}}(\ket{\psi})
\end{equation}
has a non-trivial upper bound. We find that the integer solutions, $I_{A_\gamma,B_\gamma}(\ket{\psi}) \in \mathbb{N} \quad \mathrm{for\,all\,}\gamma$, are exactly the extremal points of these two inequalities. We show explicitly that all solutions, given as four-vectors, correspond to the $60$ possible two-qubit encoded stabilizer states, and are thus are stabilizer states.

Combined with Theorem~\ref{thrm:lrm}, we finally obtain a necessary and sufficient condition, confirming that our approach is able to provide a full classification of LRN in the TC ground space.

\begin{figure}
\includegraphics[width=\linewidth]{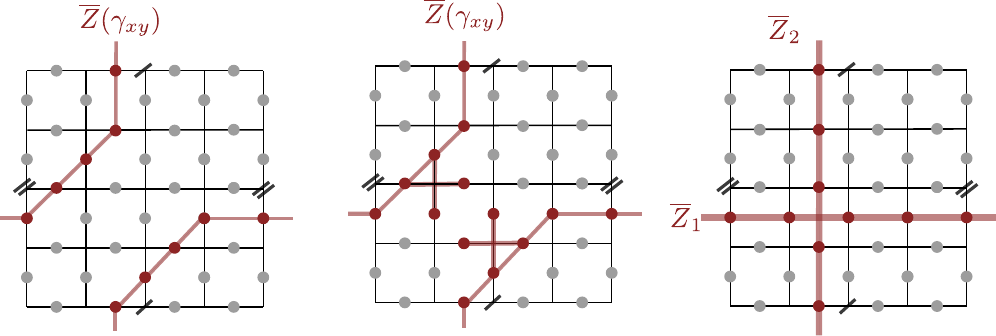}
\caption{Left: support of the  WLO $\overline{Z}(\gamma_{xy})$ on the dual lattice, corresponding to the loop $\gamma_{xy}$ in Fig.~\ref{fig:ttt}. Center, Right: by applying vertex operators, we can change the support of the WLO, mapping it to the product $\overline{Z}_1 \overline{Z}_2$ defined in Fig.~\ref{fig:tc1}. This procedure does not change its eigenstates within the ground space. Similar manipulations can be done for the WLO $\overline{X}(\gamma_{xy})$.}
\label{fig:xy_lattice}
\end{figure}

\section{Non-abelian topological order}
\label{sec:non_abelian}

In this section, we extend our approach to systems with non-abelian topological order. For concreteness, we study string-net models on the torus and use the doubled Fibonacci model as our main example. We find that the mutual information alone has limited power in diagnosing the LRN in general non-Abelian theories.

\subsection{The string-net models}

\begin{figure}
\centering
\includegraphics[width=0.8\linewidth]{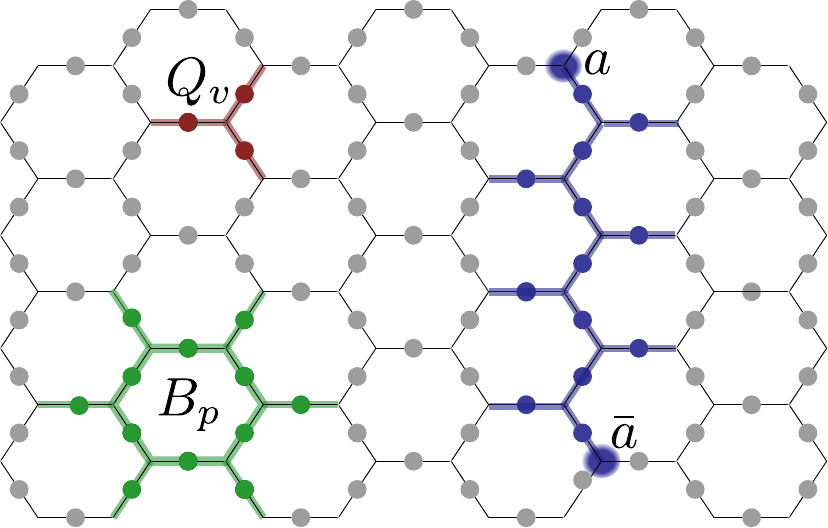}
\caption{Pictorial representation of the string-net lattice. The figure also shows the supports of the vertex, plaquette, and string operators.}
\label{fig:placeholder}
\end{figure}

The string-net models, introduced in Ref.~\cite{levin2005string}, form a family of models that, under the assumption of a ``strict" area law, describe all possible two-dimensional, gapped phases of matter~\cite{kim2024classifying}. Each element in this family admits a lattice realization, whose details depend on a set of input data, as we now briefly review.

First, string-net models are defined on a honeycomb lattice, with local degrees of freedom placed on the edges $e\in E$, cf. Fig.~\ref{fig:placeholder}. The local Hilbert space dimension associated with each edge depends on the number of ``string types'', which is the first piece of input in the model. In the case of the doubled Fibonacci model, $N_a=2$, so that each edge hosts a qubit\footnote{For more general string-net models, $N_a>2$. If one insists on considering qubits, rather than qudits, then the qudit can be decomposed in terms of qubits. In this work, we restrict to string-net models defined on qubit lattices.}. 

In addition to the number of string types, the model is specified by a set of allowed \emph{branching} (or fusion) rules\footnote{These fusion rules should not be confused with the fusion rules of the anyons. The input data of the string-net model is a unitary fusion category, while the output is a unitary modular tensor category, which has a notion of braiding.}. Denoting the fusion of string types by $a\times b$, the branching rules read
 \begin{equation}
\label{eq:fusion_rule}
a\times b =\sum_{c} N_{a,b}^c c\,,
\end{equation}
where $N_{a,b}^c$ are complex numbers satisfying some consistency conditions~\cite{simon2023topological}. For the doubled Fibonacci model, the branching rules can be found in, e.g., Refs.~\cite{levin2005string, simon2023topological}.

The string types and branching rules determine the lattice Hamiltonian
\begin{equation}
\label{eq:string_net_hamiltonian}
H=-\sum_v Q_v-\sum_p B_p\,.
\end{equation}
Here, the sums run over all the vertices $v$ and plaquettes $p$, the operator $Q_v$ acts on the three qubits connected to the  vertex $v$, while $B_{p}$ is a twelve-qubit plaquette operator, as depicted in Fig.~\ref{fig:placeholder}. Indeed, while the structure of the Hamiltonian in Eq.~\eqref{eq:string_net_hamiltonian} is common to all string-net models, the action of the operators $Q_v$ and $B_p$ depends on the fusion rules. As we will not need their explicit form, we refer to Ref.~\cite{levin2005string,Lin2021Generalized} for more detail and explicit expressions in the case of the doubled Fibonacci model. We note that the $\ZZ_2$ TC is a special case of the  string-net model, with the vertex and plaquette operators replaced with unitaries, as opposed to projectors.

In general, the string-net models exhibit anyonic excitations. In the doubled Fibonacci model, for example, the anyon theory $\mathcal{A}$ contains four elements, denoted by $I$ (the trivial anyon), $\tau$, $\bar\tau$, and $\tau\bar{\tau}$. On a torus geometry, the degeneracy of the ground space of the Hamiltonian coincides with the number of anyon types in the model, consistent with TQFT predictions~\cite{simon2023topological}. 

\subsection{The MES and the mutual information}

We now generalize Theorem~\ref{thrm:lrm} to the present setting. Our derivation closely follows the one given for the $\mathbb{Z}_2$ TC. It is based on the properties of MES, which can be defined for the string-net models as follows.

First, given a path $\gamma$ on the lattice, we recall that one can construct a string operator $W_{a}(\gamma)$ that creates a quasiparticle $a$ at the string's starting point, and the corresponding anti-particle $\bar{a}$ at the string endpoint~\cite{levin2005string,Lin2021Generalized}. Different from the TC, the operator $W_{a}(\gamma)$ is not only supported on $\gamma$, but also on a set of of neighboring edges, cf. Fig.~\ref{fig:placeholder}. In addition, $W_{a}(\gamma)$ are typically not unitary for non-abelian anyons.

Next, let us consider a non-contractible loop $\gamma$. The restriction of the operators $\{W_{a}(\gamma)\}_a$ to the ground space commute. To see this, note that the action of $W_{a}(\gamma)W_{b}(\gamma)$ on the ground space equals the action of $W_{a}(\gamma)W_{b}(\gamma')$ where $\gamma'$ is a topologically equivalent loop, sufficiently separated from $\gamma$. This is because the action of string operators on the ground space is, by construction, the same for all topologically equivalent paths~\cite{levin2005string,Lin2021Generalized}. As a result, all $W_{a}(\gamma)$ commute on the ground space, and they can be simultaneously diagonalized. We call the set of common eigenstates the MES, and denote them by $\{\ket{a;\gamma}\}_a$. This is an orthonormal basis of the ground-state subspace where each state is labeled by a unique set of eigenvalues of the $W_a(\gamma)$, determining the anyon charge. 

 We argue that the MES in the string-net models share the same properties of the MES in the TC. Consider, for instance, a loop $\gamma_y$ around the $y$ direction and take the expansion of a ground state as in Eq.~\eqref{eq:generic_state}. We claim that the reduced density matrix over a region $A_{y}$ containing $\gamma_y$ is block-diagonal, as in Eq.~\eqref{eq:rho_blocks}. Intuitively, this is true because the WLO mapping one MES onto another cannot be supported on $A_y$. We now give a formal proof.

This claim follows if, for all Pauli strings $P\in A_y$, we have
\begin{equation}
\label{eq:to_prove_1_string_net}
\braket{a;\gamma_y|P| b;\gamma_y} = 0\,, \quad \mathrm{for}\, a\neq b.
\end{equation}
The validity of Eq.~\eqref{eq:to_prove_1_string_net} is less obvious compared to the TC case, because $P$ does not necessarily map Hamiltonian eigenstates into one another. Still, as we show in Appendix~\ref{sec:technical details}, one can adapt the previous reasoning and arrive at the following expression
\begin{equation}
\label{eq:decomposition_to_prove_1}
    P\ket{b;\gamma_y}=\alpha_{\rm ex}\ket{v}_{\rm ex}+ \sum_{d}  \alpha_dW_d(\gamma_y)\ket{b;\gamma_y}\,,
\end{equation}
where $\ket{v}_{\rm ex}$ is a sum of excited states. Note that for the identity particle $a=I$ the WLO is $W_I(\gamma_y) = \mathbbm{1}$. Now, because $a\neq b$ and $W_d(\gamma_y)$ acts diagonally in the MES basis, Eq.~\eqref{eq:to_prove_1_string_net} follows.

We can extend the TC argument in a similar way to prove the factorization property [Eq.~\eqref{eq:factorization}] of the reduced density matrix associated to a MES. Here, one also uses the fact that the excitations created by $P_{A_y}$ are localized around $A_y$, and cannot be annihilated by $P_{B_y}$. Therefore, given an arbitrary ground state $\ket{\psi}$, we obtain the following form for the reduced density matrix 
\be\label{eq:final_decomposition_rho_string_net}
\rho_{A_\gamma B_\gamma} = \bigoplus_{a\in \mathcal{A}} |\psi_a(\gamma)|^2 \rho^a_{A_\gamma} \otimes \rho^a_{B_\gamma}.
\ee
Here, $\gamma$ is a non-contractible loop, while $A_\gamma$ and $B_\gamma$ are disjoint regions that are sufficiently large to support the lattice WLO $W_{a}(\gamma)$.  

Eq.~\eqref{eq:final_decomposition_rho_string_net} is all we need to extend our first main result to string-net models. That is, that the mutual information between well-separated non-contractible regions is sufficient for demonstrating that there is LRN. Indeed, it allows us to reproduce the proof of Lemma~\ref{main_lemma} without modifications, leading to a direct generalization of Theorem~\ref{thrm:lrm}.

\subsection{Constraints on LRN in string-net models from modular transformations}
\label{sec:full_classification_string_nets}

Finally, we explore the extent to which our second result, Theorem~\ref{thrm:srm_tc}, generalizes to string-net models. 

We begin by considering a non-contractible loop $\gamma_y$ along the vertical direction, and two separated, disjoint regions $A$ and $B$ supporting paths on the lattice that are equivalent to $\gamma_y$. Given a ground state $\ket{\psi}$, from the results of the previous section, we have
\begin{equation}\label{eq:mutual_info_gamma_y_1}
I_{A,B}(\ket{\psi}) = H(\{|\psi_{a}(\gamma_y)|^2\}),
\end{equation}
where $\psi_{a}(\gamma_y)$ are the coefficients associated to $\ket{\psi}$ in the MES basis $\{\ket{a;\gamma_y}\}_{a}$. We now ask how the mutual information is modified if we choose disjoint regions corresponding to different non-contractible loops. To answer this question, we need to discuss the modular transformations, already introduced in Sec.~\ref{sec:full_classification_TC}.

We recall that, in a TQFT, any modular transformation $\zeta$ on the torus induces a unitary transformation $M_\zeta$ on the degenerate ground space~\cite{kitaev2006anyons,simon2023topological}. As mentioned, the set of modular transformations is generated by the $\pi/2$ rotation $s$ and Dehn twist $t$, which correspond to the unitary operations $M_t$ and $M_s$, acting on the ground space. The defining property of this representation is that for a modular transformation $\zeta$ the WLO transform as 
\be\label{eq:modular_trafo_wlo}
M_\zeta W_a(\gamma) M_\zeta^\dagger = W_a(\zeta(\gamma)).
\ee
\begin{figure}
    \centering
    \includegraphics[width=1\linewidth]{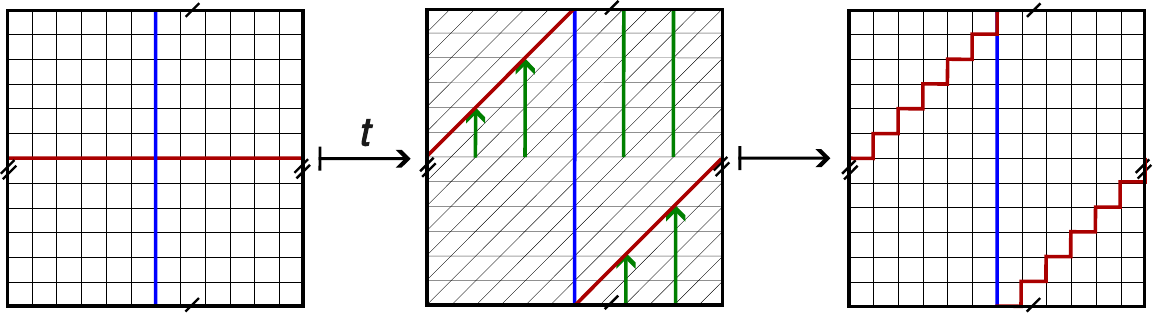}
    \caption{Dehn twist on the torus implemented by a shear transformation (green arrows), followed by a local geometric deformation to restore the lattice. The coordinate grid represents the connectivity of the lattice. Note that the shear transformation is a permutation where in the $n$-th column the vertices are shifted by $n$ steps.}
    \label{fig:dehn_twist_lattice}
\end{figure}

Unfortunately, a modular transformation $\zeta$ typically changes the connectivity of a lattice defined on the torus, and hence of the corresponding microscopic Hamiltonian. Therefore, $\zeta$ is typically not an automorphism of the lattice model, as is evident if one considers the lattice permutation implementing the shear of the lattice associated to $t$~\cite{zhu2020instantaneous,smith2020intrinsic}, which we sketch in Fig.~\ref{fig:dehn_twist_lattice}.

In order to gain an intuition, let us first neglect this fact, and make the simplifying  assumption that $\zeta$ is an automorphism of the lattice. This is for example the case for a $2\pi/3$ rotation on the honeycomb lattice. Let $\tilde\gamma$ be any non-contractible loop that is topologically inequivalent to $\gamma_y$, and let $\zeta$ be a modular transformation mapping $\tilde\gamma$ to $\gamma_y$.\footnote{The most general $\tilde\gamma$ is a curve winding $m$-times the $x$-direction and $n$-times the $y-$direction with $\gcd(m,n)=1$. This condition is equal to the property that $\tilde\gamma$ can be deformed into a non-self-intersecting curve. Such a non-contractible curve is called a simple curve and modular transformations map simple curves to simple curves \cite{farb2011primer,beverland2016protected,zhu2020instantaneous}} Under the assumption that $\zeta$ is an automorphism of the lattice, there exists a permutation of the lattice $\mathcal{P}_\zeta$\footnote{In the notation we don't distinguish between the permutation and its unitary representation on the Hilbert space} implementing $\zeta$, such that
\begin{equation}\label{eq:wrong_assumtion}
\mathcal{P}_\zeta|\psi\rangle =M_\zeta|\psi\rangle\,,
\end{equation}
for a ground state $\ket{\psi}$. Below, we explore the consequences of this assumption.
\begin{figure}
    \centering
\includegraphics[width=0.9\linewidth]{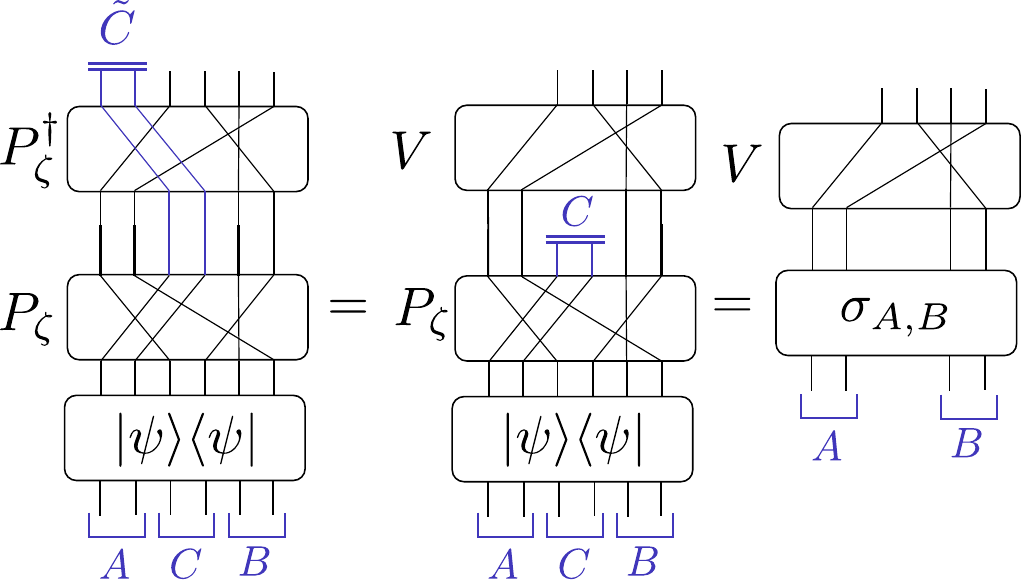}
    \caption{Pictorial derivation of Eq.~\eqref{eq:rho_barAbarB}. We only show the ``ket" layer. The horizontal double lines (colored) indicate a trace over $\tilde{C}$ on the left, and over $C$ in the middle pannel.}
    \label{fig:tn_trace}
\end{figure}

First, the inverse permutation $\mathcal{P}^{-1}_\zeta$ maps the regions $A$, $B$ (containing $\gamma_y$) into regions $\tilde A$, $\tilde B$ containing curves that are topologically equivalent to $\tilde \gamma$. Therefore, denoting by $C$ ($\tilde{C}$) the complement of $A\cup B$ ($\tilde A \cup \tilde B$), we have
\begin{align}\label{eq:rho_barAbarB}
\rho_{\tilde A\tilde B}&={\rm Tr}_{\tilde {C}}[\ket{\psi}\bra{\psi}]\nonumber\\
&={\rm Tr}_{\tilde {C}}\left[\mathcal{P}_\zeta^\dagger \mathcal{P}_\zeta\ket{\psi}\bra{\psi}\mathcal{P}^\dagger_\zeta\mathcal{P}_\zeta\right]\nonumber\\
&= V\left({\rm Tr}_C[\mathcal{P}_{\zeta}\ket{\psi}\bra{\psi}\mathcal{P}^\dagger_{\zeta}]\right)V^\dagger\,.
\end{align}
This is shown graphically in Fig.~\ref{fig:tn_trace}. Here, because $P_\zeta$ is a permutation, $V =V_{\tilde{A}\leftarrow A}\otimes V_{\tilde{B}\leftarrow B}$ is an isometry of the form
\be V:\HH_A\otimes \HH_B \to \HH_{\tilde{A}}\otimes \HH_{\tilde{B}}\ee 
with $\HH_R$ denoting the Hilbert space associated with the qubits of region $R$. Thus, setting
\be
\rho_{\tilde A\tilde B} = V\sigma_{AB} V^\dagger
\ee
it follows that the mutual information of $\rho_{\tilde A\tilde B}$  between regions $\tilde{A},\tilde{B}$ equals that of $\sigma_{AB}$ between regions $A,B$, \textit{i.e.}, we have shown
\be
I_{P_\zeta^{-1}(A),P_\zeta^{-1}(B)}(\ket{\psi}) = I_{\tilde{A},\tilde{B}}(\ket{\psi}) = I_{A,B}(P_\zeta\ket{\psi}).
\ee

From Eqs.~\eqref{eq:mutual_info_gamma_y_1} and~\eqref{eq:wrong_assumtion}, we arrive at
\begin{equation}\label{eq:final_result_i_transform}
I_{\tilde A,\tilde B}(\ket{\psi}) = H(\{|\tilde{\psi}_{a}|^2\})\,,
\end{equation}
where 
\begin{equation}\label{eq:final_transform}
    \tilde{\psi}_{b}= \sum_b(M_\zeta)_{b,a} \psi_{a}(\gamma_y)\,.
\end{equation}
Note that $\tilde{\psi}_{b} = \psi_b(\tilde{\gamma})$ are the coefficients of $\ket{\psi}$ in the basis $\ket{a;\tilde{\gamma}}$, as dictated by Eq.~\eqref{eq:modular_trafo_wlo}.
In words, Eq.~\eqref{eq:final_result_i_transform} states that the mutual information with respect to $\tilde{A}$ and $\tilde{B}$ is the classical Shannon entropy of the coefficients $\{|\tilde{\psi}_{a}|^2\}_a$, where $\tilde{\psi}_{a}$ is obtained by applying $M_\zeta$ to the vector $\{\psi_{b}(\gamma_y)\}_b$. 

In turn, this would lead us to the following conclusion: an encoded state corresponding to the vector $\{\psi_a(\gamma_y)\}$ has LRN if the Shannon entropy of the coefficients
$|\sum_a(M_\zeta)_{b,a} \psi_a(\gamma_y)|^2$ is not an integer for at least one modular transformation. Indeed, in this case it means that we can find two regions $\tilde{A}$ and $\tilde{B}$, with respect to which the mutual information of our state $\ket{\psi}$ is not an integer.

If $\zeta$ is not an automorphism of the lattice, the assumption in Eq.~\eqref{eq:wrong_assumtion} is not correct. However, we may arrive at the same conclusion by making use of the exact lattice  implementation of the modular transformations in string-net models~\cite{zhu2020instantaneous,zhu2020quantum,smith2020intrinsic}.  Specifically, given a modular transformation $\zeta$, Ref.~\cite{zhu2020instantaneous} provides an explicit construction for a unitary operation on the full Hilbert space that implements $M_\zeta$ when restricted to the ground space. The construction requires introducing an ancillary lattice made of $N_\alpha$ sites, whose precise geometry is irrelevant for our discussion. Then, one of the results of Ref.~\cite{zhu2020instantaneous} implies the following statement: for any $\zeta$, there exists a permutation operator $\mathcal{P}_\zeta$ and a finite-depth local unitary circuit $\mathcal{U}_\zeta$, such that\footnote{In the original construction of Ref.~\cite{zhu2020instantaneous}, the order of $\mathcal{P}_\zeta$ and $\mathcal{U}_\zeta$ is exchanged. However, as stated there, one can reverse the order by modifying the circuit. We use the ordering of Eq.~\eqref{eq:correct_implementation} as it simplifies our technical analysis.}
\begin{equation}
\mathcal{U}_\zeta\mathcal{P}_\zeta(|\psi\rangle\otimes \ket{0}^{N_\alpha} )=\left(M_{\zeta}|\psi\rangle\right)\otimes \ket{0}^{N_\alpha}\,,
\label{eq:correct_implementation}
\end{equation}
where $\ket{0}^{N_\alpha}$ denotes a product state in the ancilla Hilbert space. Note that $\mathcal{U}_\zeta$ and $\mathcal{P}_\zeta$ now both act on the composite lattice made of both the original and ancillary qubits. However, the composite action of $\mathcal{U}_\zeta$ and $\mathcal{P}_\zeta$ leave the ancillas in a product state.

Eq.~\eqref{eq:correct_implementation} replaces the assumption in Eq.~\eqref{eq:wrong_assumtion} in general. It features an additional local operator $\mathcal{U}_\zeta$, which compensates for the lattice deformation caused by the permutation associated with the geometrical modular transformation on the torus, cf.~Fig.~\ref{fig:dehn_twist_lattice}. As we show in Appendix~\ref{sec:proof_mutual_info_change_stringnets}, we can now repeat the previous arguments, up to mild additional technical complications. As a result, we arrive at our second main result characterizing LRN in string-net models, that we state here in the form of a theorem.
\begin{thm}\label{thrm:classification_string_nets}
    Let $\ket{\psi}$ be a ground state of the string-net Hamiltonian in Eq.~\eqref{eq:string_net_hamiltonian}. Fix the MES basis $\{\ket{a;\gamma_y}\}_a$, and denote by $\{\psi_a(\gamma_y)\}_a$ the coefficients of $\ket{\psi}$ in this basis. Denote by $M_\zeta$ a modular operator, obtained as an arbitrary product of $T$ and $S$ matrices, and denote by $\psi^\zeta_a=\sum_b (M_\zeta)_{a,b}\psi_b(\gamma_y)$. Then, if 
    \begin{equation}
       H(\{|\psi_a^\zeta|^2\})\notin \mathbb{N} 
    \end{equation}
    for some $\zeta$, $\ket{\psi}$ has LRN.
\end{thm}

As a final remark, we emphasize that, while the depth of the circuit $\mathcal{U}_\zeta$ is finite and independent of $N$, in general it depends on the modular transformation. Therefore, for a given $\zeta$, the system should be considered to be ``large enough'', namely large compared to the depth of the circuit. Because LRN is an asymptotic property, this is not a restriction. 

\subsection{Implications for the Fibonacci string-net model}

In this section, we discuss the consequences of Theorem~\ref{thrm:classification_string_nets} for the doubled Fibonacci string-net model. As we will show, contrary to the $\mathbb{Z}_2$ TC, we will not be able to fully classify LRN. While we can establish LRN for typical states and in particular the MES, we will also obtain families of states with integer mutual information for all allowed regions, although LRN is expected. This fact manifests a limitation of our elementary approach based on the mutual information. 

We begin by fixing the MES basis 
\begin{equation}\label{eq:b_y}
    \mathcal{B}_y=\{|I;\gamma_y\rangle, |\tau;\gamma_y\rangle,|\bar \tau;\gamma_y\rangle,|\tau \bar \tau;\gamma_y\rangle\}\,.
\end{equation}
The $S$ and $T$ matrices of the model read~\cite{simon2023topological}
\be
S
=
\frac{1}{1+\phi^2}
\begin{pmatrix}
1 & \phi & \phi & \phi^2 \\
\phi & -1 & \phi^2 & -\phi \\
\phi & \phi^2 & -1 & -\phi \\
\phi^2 & -\phi & -\phi & 1
\end{pmatrix}\,,
\ee  
where $\phi=(1+\sqrt{5})/2$, and 
\be 
T=
\begin{pmatrix}
1 & 0 & 0 & 0 \\
0 & e^{\frac{4\pi i}{5}} & 0 & 0 \\
0 & 0 & e^{-\frac{4\pi i}{5}} & 0 \\
0 & 0 & 0 & 1
\end{pmatrix}\,.
\ee
The MES correspond to vectors $\psi_a(\gamma_y) = \delta_{a,b}$. Transforming to the $\gamma_x$-basis using the $S$ matrix, we obtain $\psi_a(\gamma_x)$ as the columns of the $S$ matrix. For those vectors, we find that $H(\{|\psi_a(\gamma_x)|^2\})$ is not integer valued and we conclude LRN for the MES. Consider now superpositions of the form 
\be
\alpha \ket{I;\gamma_y} + \beta \ket{\tau;\gamma_y}
\ee
corresponding to $\psi_a(\gamma_y) = (\alpha,\beta,0,0)$. This family of states has LRN: For choices of $\alpha,\beta$ such that $H(\{|\psi_a(\gamma_x)|^2\})$ is integer valued, we find that under transformation of $S$ or $ST$ the quantization is not preserved.

Let us now discuss cases where our approach is inconclusive.
First, it can be checked that the state
\be 
\ket{v_1}=\frac{1}{\sqrt{2}} \left(\ket{I;\gamma_y} + \ket{\tau\bar{\tau};\gamma_y}\right)\,,
\ee 
is a simultaneous eigenstate of $S$ and $T$ with eigenvalue one. In fact, one can show that it is the only common eigenstate. Going further, we define the following states 
\begin{align}
\ket{v_2}&=\frac{1}{\sqrt{2}} \left(\ket{\tau;\gamma_y} + \ket{\bar{\tau};\gamma_y}\right)\,, \\
\ket{v_3}&=\frac{i}{\sqrt{2}} \left(\ket{\tau;\gamma_y} - \ket{\bar{\tau};\gamma_y}\right)\,,\\
\ket{v_4}&=\frac{1}{\sqrt{2}} \left(\ket{I;\gamma_y} - \ket{\tau\bar{\tau};\gamma_y}\right)\,,
\end{align}
which, together with $\ket{v_1}$, form an orthogonal basis of the ground space. 
The restriction of $S$ and $T$ to the subspace $K={\rm span}(\ket{v_2},\ket{v_3},\ket{v_4})$ yields
\be
S\big|_{K} =
\begin{pmatrix}
\dfrac{1}{\sqrt{5}} & 0 & \dfrac{2}{\sqrt{5}} \\
0 & -1 & 0 \\
\dfrac{2}{\sqrt{5}} & 0 & -\dfrac{1}{\sqrt{5}}
\end{pmatrix}\,,
\ee
and
\be
T\big|_{K}  = 
\begin{pmatrix}
\cos(4\pi/5) & \sin(4\pi/5) & 0 \\
-\sin(4\pi/5)  & \cos(4\pi/5) & 0 \\
0 & 0 & 1
\end{pmatrix}\,.
\ee

Since $\ket{v_1}$ is fixed by the modular group, the mutual information between disconnected regions is an integer for all choices of non-contractible loops. Therefore, our approach is not able to detect LRN. Still, one could hope that the restricted matrices $T\big|_{K}$ and $S\big|_{K}$ generate a dense subset of matrices, so that one could rule out the existence of integer mutual information for states with $|v_1|\neq 1$. Unfortunately, however, this is not the case. Indeed, one can check that the group generated by $T\big|_{K}$ and $S\big|_{K}$ contains $60$ elements and is thus finite, as they form a presentation for the Icosahedral symmetry~\cite{freedman2001two}. The modular matrices generating a dense set would require a different manifold than the torus~\cite{freedman2001two,ng2010congruence}. 

Going further, consider the orbit of the following state under modular transformations:
\begin{equation}
\label{eq:w_states}
    \ket{w}=\frac{1}{2}(\ket{I}+\omega^2 |\tau\rangle +\omega^2 |\bar \tau\rangle +\omega \ket{\tau \bar\tau} )\,,
\end{equation}
where $\omega$ is a sixth root of unity, namely $\omega^6=1$. While $\ket{w}$ is not invariant under the modular group, it is possible to check that all twelve of the states in its orbit have integer mutual information.

In summary, in the double Fibonacci model, the mutual information is not able to fully classify LRN. It is worth noting, however, that while the states $\{M\ket{w}\}_M$, where $M$ is a modular matrix and $\ket{w}$ is from Eq.~\eqref{eq:w_states}, span the whole encoded subspace, they are not orthogonal to one another. Therefore, one could wonder whether it is possible to find a orthogonal basis of the encoded space such that each element has quantized mutual information. We conjecture that this is not possible, although we were not able to prove it. If the conjecture were true, then one could establish LRN of any ground state using the robust ``infectiousness" property proven in Refs.~\cite{wei2025long,parham2025quantum}. This theorem states that if at least one ground state of a local Hamiltonian with topological order (see Refs.~\cite{wei2025long,parham2025quantum} for the technical assumptions, which apply to string-net models) has short-range nonstabilizerness, then the full ground-state subspace is mapped to a topological stabilizer code~\cite{bravyi2013classification,bravyi2009nogo,ellison2022pauli} via a shallow-depth local quantum circuit. In particular, this implies that the ground-state subspace has a basis of short-range nonstabilizer states with integer mutual information on regions such as in Fig.~\ref{fig:ttt}.

\section{Discussions}\label{sec:open_ends}

We expect the results derived for the $\ZZ_2$ TC and the doubled Fibonacci string-net model to extend to a broad class of anyon models. More generally, however, subtle details may arise, requiring dedicated analyses. For instance, in the abelian $\ZZ_d$ TC, the  $\ZZ_d$ stabilizers may be intricate for non-prime local dimension $d$~\cite{ellison2022pauli}. It is an interesting question whether our approach can completely classify LRN in this case.

We would also like to point out that LRN is tied to the notion of fault tolerant gates; specifically,~logical gates induced by a shallow-depth local quantum circuit (QC) \cite{bravyi2013classification}. For example, knowing that the $T$ state in the $
\ZZ_2$ toric code has LRN implies that the $T$ gate cannot be realized fault tolerantly. Explicitly, consider the $T$ state encoded in the first logical qubit from Eq.\eqref{eq: T state}:
\begin{align}    
\ket{T}_L &= \frac{1}{\sqrt{2}} \left( \ket{00}_L + \exp(i\pi/4) \ket{10}_L\right)\,.\\
&= T_L\ket{+,0}_L
\end{align}
Since $\ket{+,0}_L$ is a stabilizer state and we proved that $\ket{T}_L$ has LRN, the logical $T$ gate cannot be realized in shallow-depth with a local QC, not even approximately.

In Theorem~\ref{thrm:srm_tc} we showed that all short-range non-stabilizer states in the $
\ZZ_2$ TC are stabilizer states. Thus, using the same argument, we obtain that if a logical gate is implemented with a shallow-depth circuit, it must be a Clifford unitary operator. This result coincides with the Bravyi-K\"onig Theorem~\cite{bravyi2013classification}.

Furthermore, we can characterize fault tolerant gates more systematically, so as to make a statement that applies to any two-dimensional topological order. Indeed, we have shown that a shallow-depth local QC leaves the mutual information $I_{A_\gamma,B_\gamma}[\ket{\psi}] = H(\{|\psi_a(\gamma)|^2\})$ of a ground state $\ket{\psi}$ invariant. Since this is true for any ground state $\ket{\psi}$ and any choice of $\gamma$, on the ground-state subspace (code space), the QC can only induce a diagonal phase gate composed with a permutation $\sigma$ of the anyon labels:
\be\label{eq:ftg}
\psi_a(\gamma) \to e^{i\phi_a}\psi_{\sigma(a)}(\gamma).
\ee

Indeed, such a unitary does not change the value of the Shannon entropy $H(\{|\psi_a(\gamma)|^2\})$. Conversely, if for all $\psi_a(
\gamma)$ a unitary $U$ mapping ground states to ground states satisfies
\begin{align}
H(\{|\psi_a(\gamma)|^2\}) &= H(\{
|\phi_a(\gamma)|^2\}), \\
\phi_a(\gamma)&=\sum_{b}U_{a,b}\psi_b(\gamma),
\end{align}
 then choosing $\psi_a(
\gamma) = 
\delta_{a,b}$ implies that the $b$-th column of $U$ has a single non-zero entry of magnitude one. Since this must hold for any column and $U$ is a unitary, Eq.\eqref{eq:ftg} follows. This is the (central result of the) classification  of fault-tolerant gates of TQFTs from Ref.~\cite{beverland2016protected}. 

The ramifications of Eq.~\eqref{eq:ftg} depend on the specific modular matrices of a given TQFT~\cite{beverland2016protected}. For example, for the TC, Ref.~\cite{beverland2016protected} proves that only a subset of Clifford unitaries are allowed, while for the doubled Fibonacci phase any consistent solution of Eq.~\eqref{eq:ftg} is proportional to the identity \cite{li2026explicit}, \emph{i.e.}, there are no non-trivial fault-tolerant gates.

The proof in Ref.~\cite{beverland2016protected} shares similarities to
our mutual-information based approach, but includes additional assumptions. In particular, it is assumes that the shallow QC preserves the ground-state subspace exactly. As we can allow for an infidelity that goes to zero in the thermodynamic limit (as in the proof of Theorem~\ref{thrm:lrm}), we have proven that, on the torus, the classification in Ref.~\cite{beverland2016protected} is robust. Additionally, our results are entirely formulated on the lattice, as opposed to the continuum. 

Finally, the proof in Ref.~\cite{beverland2016protected} does not, as stated by the authors, generalize to higher dimensions. The invariance of the mutual information, on the other hand, can give constraints on fault-tolerant gates in higher dimensional topological codes. It may also be valuable for assessing the allowable gates in topological codes with defects or various boundary conditions. We leave the investigation of this direction for future work.

\section{Conclusions and Outlook}
\label{sec:outlook}

In this work, we considered an elementary approach for detecting and classifying LRN in the ground space of topologically ordered Hamiltonians based on mutual information. In particular, we considered the mutual information between non-overlapping regions containing non-contractible loops and studied the change of the mutual information under modular real-space transformations. We have exemplified our approach in the abelian toric code and the non-abelian string-net model with Fibonacci anyons. In the former case, it provides a full classification, certifying LRN for all encoded non-stabilizer states. In the latter case, instead, our approach is able to detect LRN for all states except for a finite subset with special properties under the modular group. In addition, we have discussed how our results constrain fault-tolerant gates, replicating the classification from Ref.~\cite{beverland2016protected} for the torus geometry.

Our work has implications beyond string-net models, since any ground state of a Hamiltonian in a topological phase can be mapped, under some technical assumptions, to a string-net ground state with a shallow QC \cite{kim2024classifying}. Thus, such a state has LRN if and only if the corresponding string-net state has LRN. As a consequence, our main result, Theorem~\ref{thrm:classification_string_nets}, applies to any state in a topological phase and its implications depend only on the $S$ and $T$ matrices.

Next, while we have discussed only the torus geometry and $\ZZ_2$ stabilizers, the calculation of the mutual information generalizes to different topologies, including those with defects, and in higher dimensions \cite{jian2015long}. Therefore, an extension of the definitions and theorems to different topologies and the $\ZZ_d$ stabilizer formalism should be possible.  On the other hand, in higher dimensions the implications are likely more subtle. For example, unlike in two spatial dimensions, there are topological stabilizer codes that admit fault tolerant (transversal) logical non-Clifford gates. Thus, acting with such a gate on a stabilizer state yields a non-stabilizer state with short-range nonstabilizerness~\cite{bombin2015gaugecolorcodesoptimal,Kubica2015transversalgatesincolorcodes,bombin2018transversalgateserrorpropagation, Vasmer2019CCZ}.

Finally, one could study implications of our results for the ``magic hierachy" \cite{parham2025quantum,li2026explicit} and it is natural to ask whether our discussion could be framed in terms of the projected-entangled-pair-state (PEPS) representation of the ground-state wave functions~\cite{verstraete2004renormalization,buerschaper2009explicit}. It would be especially interesting, for example, if the LRN could be related to non-standard (generalized) symmetries of the PEPS \cite{bultinck2017anyons}.
We leave these questions for future research.\\

{\bf Note added}. During the completion of this draft, the pre-print Ref.~\cite{zhang2026extensive} appeared on the arXiv. Based on the entanglement bootstrap axioms, it proves rigorously that the ground-space of non-abelian string-net models display LRN. Our approach is different and is based exclusively on the analysis of the mutual information between disconnected regions. In particular, it also classifies LRN in the ground space of the $\ZZ_2$-TC.}

\section*{Acknowledgments}
We acknowledge the use of AI tools (ChatGPT) to find the example in Eq.~\eqref{eq:w_states}. We thank Natalie Parham, Norbert Schuch, Angelo Lucia, Christopher D. White and Riccardo Cioli for discussions. TDE thanks Yuzhen Zhang and Daniel Ranard for inspiring conversations. TDE, LP, and DTS acknowledge hospitality from the Simons Center for Geometry and Physics, Stony Brook University, during the ``Workshop on
Quantum information dynamics and nonequilibrium quantum matter'', where part of the research for this manuscript was
performed. The work of LP and DK was funded by the European Union (ERC, QUANTHEM, 101114881). Views and opinions expressed are however those of the author(s) only and do not necessarily reflect those of the European Union or the European Research Council Executive Agency. Neither the European Union nor the granting authority can be held responsible for them.  Research at Perimeter Institute is supported in part by the Government of Canada through the Department of Innovation, Science, and Economic Development, and by the Province of Ontario through the Ministry of Colleges and Universities. \\

\appendix

\section{Invariance of mutual information}\label{sec:lemmas_mut_info}

In this appendix, we use the same notations as in the main text. Further, for a region $A_\gamma$ containing a non-contractible loop $\gamma$, we define the width of $A_\gamma$ as the length of the shortest path connecting its two distinct boundaries, and for two such regions $A_\gamma,B_\gamma$ the distance $\mathrm{dist}(A_\gamma,B_\gamma)$ is the length of the shortest path connecting $A_\gamma$ and $B_\gamma$.
\begin{lem}\label{lemma:mut_info_invariance}
    Let $A_\gamma$, $B_\gamma$ be two disjoint regions containing curves topologically equivalent to a non-contractible curve $\gamma$. Suppose $Q$ is a shallow QC of depth $D_Q = O(\mathrm{polylog}(N))$ and $\mathrm{width}(A_\gamma)  = \mathrm{width}(B_\gamma) \geq O(2D_Q)$. Suppose also that the complement $C$ of $A\cup B$ is such that $\mathrm{dist}(A_\gamma,B_\gamma)= O(L)\gg \mathrm{width}(A_\gamma)$, $L={\rm min}(L_x,L_y)$ and $L_x$, $L_y$ are the linear dimensions of the system, namely $N=L_x L_y$. For any state $\ket{\psi}$ such that the reduced density matrix satisfies
    \be
\rho_{A_{\gamma
} B_\gamma} = \bigoplus_a |\psi_{a}(\gamma)|^2 \rho^a_{A_\gamma} \otimes \rho^a_{B_\gamma}\,,
\ee
 the mutual information of $Q\ket{\psi}$ equals that of $\ket{\psi}$
\be \label{eq:statement_of_lemma_1}
I_{A_\gamma,B_\gamma}(Q\ket{\psi}) = I_{A_\gamma,B_\gamma}(\ket{\psi}).
\ee
\end{lem}

We prove this lemma by viewing the torus as a one-dimensional system, where the regions $A_\gamma,B_\gamma$ become local (connected) intervals of the 1D chain. 

\begin{proof}

For simplicity of notation, we will prove the theorem choosing $\gamma$ along the vertical direction, $\gamma=\gamma_y$, but the generalization to an arbitrary curve is straightforward. 

Denote by $L_x$, $L_y$, the linear dimensions of the lattice, so that $N=L_x\times L_y$. We can view the torus as a $1D$ ring with Hilbert space 
\be
\mathcal{H}_{1D} = \bigotimes_{i}\left(\bigotimes_j \mathcal{H}_{i,j}\right) = \bigotimes_{i} \mathcal{H}_i'
\ee
and $\mathcal{H}'_i\simeq \mathbb{C}^{2^{L_y}}$. We will refer to this process of viewing the torus as a $1D$ ring as a ``compactification''.

In the torus viewed as a $1D$ chain, two regions $A_y,B_y$ such as in the middle panel of Fig.~\ref{fig:ttt}, are connected intervals $A'$, $B'$. Consider now the claim~\eqref{eq:statement_of_lemma_1} that for a state $\ket{\psi}$ the mutual information is invariant under a shallow circuit. Since a shallow QC remains shallow under compactification (note the inverse is not true) it is sufficient to prove the following statement: given a $1D$ state $\ket{\psi'}$, and two disjoint, sufficiently separated intervals $A'$, $B'$ such that
\be\label{eq:loc_ortho}
\rho'_{A'B'} = \bigoplus_a |\psi_{a}(y)|^2 \rho^a_{A'}\otimes \rho^a_{B'},
\ee
then for any QC $Q'$ whose depth is much smaller than the distance between $A'$ and $B'$, we have
\be\label{eq:to_prove_intermediate}
I_{A',B'}(\ket{\psi'}) = I_{A',B'}(Q'\ket{\psi'})\,.
\ee
The proof of this statement can be found in Ref.~\cite{korbany2025long}.
\end{proof}

\section{All short-range nonstabilizer states in the toric code}\label{sec:srm_tc}

In this appendix, we prove Theorem~\ref{thrm:srm_tc}. In the following, we denote a ground state of the TC in the computational logical basis by
\begin{align}
\ket{\psi} &= \alpha_1 \ket{00}_L + \alpha_2\ket{10}_L + \alpha_3 \ket{01}_L + \alpha_4 \ket{11}_L.
\end{align}
We will make use of the bases of MESs $\ket{a;\gamma}$ for $\gamma \in \{\gamma_x,\gamma_y,\gamma_{xy}\}$. It will be convenient to express the coefficients $\psi_{a}(\gamma_y) = \braket{\psi|a;\gamma_y}$ in terms of the coefficients $\alpha_i$ according to
\begin{align}\label{eq:basis_trafo_y}
\ket{\psi} &= \frac {1}{\sqrt{2}} (\alpha_1 + \alpha_2) \ket{+, 0}_L + \frac {1}{\sqrt{2}} (\alpha_1 - \alpha_2)\ket{-,0}_L  \nonumber\\ 
+& \frac {1}{\sqrt{2}} (\alpha_3 + \alpha_4) \ket{+,1}_L + \frac {1}{\sqrt{2}} (\alpha_3 - \alpha_4)\ket{-,1}_L \nonumber \\
&= \sum_a \psi_a(\gamma_y)\ket{a;\gamma_y}\,,
\end{align}
while for $\psi_{a}(\gamma_x) = \braket{\psi|a;\gamma_x}$ we have
\begin{align}\label{eq:basis_trafo_x}
\ket{\psi} &= \frac {1}{\sqrt{2}} (\alpha_1 + \alpha_3) \ket{0,+}_L + \frac {1}{\sqrt{2}} (\alpha_1 - \alpha_3)\ket{0,-}_L  \nonumber\\ 
+& \frac {1}{\sqrt{2}} (\alpha_2 + \alpha_4) \ket{1,+}_L + \frac {1}{\sqrt{2}} (\alpha_2 - \alpha_4)\ket{1,-}_L\nonumber \\
&=\sum_a \psi_a(\gamma_x)\ket{a;\gamma_x}.
\end{align}
Finally, for $\gamma_{xy}$:
\begin{align}\label{eq:basis_trafo_xy}
\ket{\psi} &= \frac {1}{\sqrt{2}} (\alpha_1  \!+ \! \alpha_4) \ket{\psi ^{+}}_L  \!+ \!\frac {1}{\sqrt{2}} (\alpha_1 \! -  \!\alpha_4)\ket{\psi ^{-}}_L  \!+ \! \frac {1}{\sqrt{2}} (\alpha_3  \!+  \!\alpha_2) \ket{\phi^{+}}_L  \!+ \! \frac {1}{\sqrt{2}} (\alpha_3  \!- \! \alpha_2)\ket{\phi ^{-}}_L\nonumber \\
&=\sum_a \psi_a(\gamma_xy)\ket{a;\gamma_{xy}}.
  \end{align}

We begin by the following observation:
\begin{lem}\label{lemma_lower_bound_tc}
 Let $\ket{\psi}$ be a ground state of the toric code and denote $|\psi_{a}(\gamma)|^2 = |\braket{\psi|a;\gamma}|^2$, then for distinct  $\gamma$ and $\gamma'$
    \begin{align}
    &I_{A_\gamma,B_\gamma}(\ket{\psi}) + I_{A_{\gamma'},B_{\gamma'}}(\ket{\psi}) 
    = H(\{|\psi_{a}(\gamma)|^2\})  + H(\{|\psi_{a}(\gamma')|^2\}) 
    \geq 2.
    \end{align}
   Furthermore, if the lower bound is saturated, then the distributions $|\psi_{a}(\gamma)|^2$ and $|\psi_{a}(\gamma')|^2$ are flat (all non-zero values are equal to each other).
\end{lem}
\begin{proof}
This follows immediately from Eq.~\eqref{eq:mut_info_Shannon} and the Maassen–Uffink bound~\cite{maassen1988generalized,coles2017entropic}, which we recall. For two orthonormal bases $
\ket{e_i}$ and $\ket{\widehat{e}_i}$ and a normalized state $\ket{\psi}$ define $c_i = \braket{e_i|\psi}$ and $\widehat{c}_i = \braket{\widehat{e}_i|\psi}$. Then 
\begin{align}
H(\{|c_i|^2\}) + H(\{|\widehat{c}_i|^2\}) \geq -2\log_2(F)\,,
\end{align}
where
\begin{equation}
 F = \max_{i,j} |\braket{e_i|\widehat{e}_j}|\leq 1.
\end{equation}
It is straightforward to apply this bound to our case. Indeed, here the change of basis is one of the modular matrices $M\in \{S,ST,STS\}$, which have entries $|M_{ab}| = \frac{1}{2}$. Finally, it is known that only flat distributions can saturate the bound \cite{abdelkhalek2015optimality}.
\end{proof}

\begin{proof}[Proof of Theorem \ref{thrm:srm_tc}]
If $\ket{\psi}$ is a stabilizer state, then for all $\gamma$ the mutual information
\be 
I_\gamma :=I_{A_\gamma,B_\gamma}[|\psi\rangle] 
\ee 
takes integer values. Thus, we only need to prove the converse statement. To this end, we will construct all solutions of $I_\gamma \in \mathbb{N}$ for all $\gamma \in \{\gamma_x,\gamma_y,\gamma_{xy}\}$ and show that all such solutions correspond to stabilizer states.

First, note that by dimensionality $I_\gamma \leq 2$, so the integer values can only be $0,1$ and $2$. We will find all solutions, by going through all possibilities of the integer values the mutual information can take.
\begin{itemize}
\item
First note that if there is one $\gamma$ such that $I_\gamma = 0$, then $\ket{\psi} = \ket{a;\gamma}$ is a MES, which is a stabilizer state by construction. While we already know that the mutual information must then take integer values in the other two bases, note that Lemma~\ref{lemma_lower_bound_tc} implies that in the other two bases $I_{\gamma'}=I_{\gamma''} =2$. Here and in the following we always denote three mutually different curves by $\gamma$, $\gamma^\prime$ and $\gamma''$.
\item Next, suppose that there are at least two bases such that $I_\gamma = I_{\gamma'}=1$. Then, Lemma \ref{lemma_lower_bound_tc} implies that up to a permutation
    \begin{align} \label{eq:flat_distr}
    \{|\psi_{a}(\gamma)|^2\} &= (1/2,1/2,0,0)
    \end{align}
    and similarly for $\gamma'$. Without loss of generality, let us assume that $\gamma=\gamma_y$ and $\gamma'=\gamma_x$, corresponding to the basis associated with the decompositions~\eqref{eq:basis_trafo_y} and~\eqref{eq:basis_trafo_x} (the other two cases are treated in the same way). First, suppose that in Eq.~\eqref{eq:basis_trafo_y} it is $|\alpha_1+\alpha_2|=|\alpha_1-\alpha_2|=0$. Then, $\alpha_1=\alpha_2=0$ and from Eq.~\eqref{eq:basis_trafo_x} we see that  and it must be either $\alpha_3=0$ and $\alpha_4\neq 0$ or $\alpha_3\neq 0$ and $\alpha_4=0$. In both cases, $\ket{\psi}$ is a product state, and clearly an encoded stabilizer state. Similarly, one can treat the case $|\alpha_3+\alpha_4|=|\alpha_3-\alpha_4|=0$. The only case left corresponds to setting $|\alpha_1-s\alpha_2|=0$ and $|\alpha_3-t\alpha_4|=0$ for $s,t=\pm 1$, which gives us  $\alpha_1 =s\alpha_2$, $\alpha_3=t\alpha_4$. Using, Eq.~\eqref{eq:basis_trafo_x} we also obtain $\alpha_2=u\alpha_4$, where $u=\pm 1$, so that putting all together, we get
\begin{equation}
\ket{\psi} \propto\left(us\ket{00}+s\ket{10} + t\ket{01} + \ket{11}\right)
\end{equation}

It is now straightforward to check that $\ket{\psi}$ is an encoded stabilizer state (and therefore a stabilizer state) for all choices of the signs $u,s$, and $t$. 
\end{itemize}
So far we have found that states such that (i) one $I_\gamma=0$, or (ii) at least two different $\gamma,\gamma'$ such that $I_\gamma = I_{\gamma'}=1$, are stabilizer states and their mutual information takes integer values in all three bases. Thus, we are left with the cases where the integer values are all two (i.e.~maximal), or $I_\gamma = 1$ and $I_{\gamma'}=I_{\gamma''}=2$. We discuss those two cases next.
\begin{itemize}
\item
Assume $I_\gamma = 2$ for all $\gamma$.
This leads to a contradiction, as follows. Using Eqs.~\eqref{eq:basis_trafo_y},~\eqref{eq:basis_trafo_x}, and~\eqref{eq:basis_trafo_xy}, the condition of maximal entropy in all bases implies that for all $i\neq j$
\be\label{eq:cond_max_entropy}
|\alpha_i \pm \alpha_j| = \frac{1}{\sqrt{2}}.
\ee
In turn, Eq.~\eqref{eq:cond_max_entropy} implies that $\{\alpha_j\}_{j=1}^4$ are four equidistant points in the complex plane. This is impossible, since the maximum number of mutually equidistant points in the plane is three (forming an equilateral triangle).

\item 
We are left with the case where there is exactly one $I_\gamma = 1$, while $I_{\gamma'} =  I_{\gamma''} =2$. This is for example the case for 
\begin{align}
    \ket{\psi} &= \frac 12\left(\ket{00}_L + i\ket{01}_L + i\ket{10}_L + \ket{11}_L\right) \nonumber \\&= \frac{1}{\sqrt{2}}\left(\ket{\psi^+}_L+ i\ket{\phi^+}_L\right),
\end{align}
where we use the notation as in Eq.~\eqref{eq:basis_trafo_xy}. This state has stabilizers $\overline{Y}_1,\overline{Y}_2$.

In Lemma \ref{lem:tc_constraints} we show that the conditions can only be satisfied for uniform distributions in the computational basis of the form
\begin{align}
    \alpha_1 &= \exp(-i\phi_2) \alpha_2 = \exp(-i\phi_3) \alpha_3 = \exp(-i\phi_4) \alpha_4 = \frac{1}{2}
\end{align}
and the relative phases are constrained to be such that $\exp(i\phi_j) = \pm 1 $ or $\exp(i\phi_j) = \pm i$. Furthermore, the phases are such that exactly two factors are imaginary. We have checked, by means of a simple computer code, that all allowed states are stabilizer states.
\end{itemize}
If we count all the states appearing in the different cases as solutions of $I_\gamma \in \mathbb{N}$ for all $\gamma$, there are $60$ different states, which exhausts all possible two qubit stabilizer states.
\end{proof}

\begin{lem}\label{lem:tc_constraints}
    Let 
    \begin{align}
\ket{\psi} &= \alpha_1 \ket{00}_L + \alpha_2\ket{10}_L + \alpha_3 \ket{01}_L + \alpha_4 \ket{11}_L,
\end{align}
be a ground state of the TC and assume $I_\gamma = I_{\gamma'} = 2$ and $I_{\gamma''} =1$. Then, $|\alpha_i|=\frac{1}{2}$ and the relative phases are such that two $\alpha_i$ are real, and two $\alpha_i$ are purely imaginary.
\end{lem}
\begin{proof}
Without loss of generality consider $I_{\gamma_x}= I_{\gamma_y} = 2$, i.e. the coefficients $|\braket{a_\gamma|\psi}|^2$ for $\gamma =\gamma_x,\gamma_y$ are uniform which implies
\begin{align}
&|\alpha_1 \pm \alpha_2|^2 = |\alpha_3 \pm \alpha_4|^2 = \frac{1}{2}, \\      
&|\alpha_1 \pm \alpha_3|^2 = |\alpha_2 \pm \alpha_4|^2 = \frac{1}{2}.
\end{align}
Let $\alpha_i = r_i\exp(\phi_i)$. Then those equations imply
\begin{align}\label{eq:constraints_amplitudes}
    &r_1^2 + r_2^2 = r_3^2 + r_4^2 =  \frac{1}{2}, \\
    &r_1^2 + r_3^2 = r_2^2 + r_4^2 =  \frac{1}{2},
\end{align}
and
\begin{align}\label{eq:constraints_phases}&\cos(\phi_1 -\phi_2) = \cos(\phi_3 -\phi_4) = 0, \\ 
    &\cos(\phi_1 -\phi_3) = \cos(\phi_2 -\phi_4) = 0 .
\end{align}
Without loss of generality consider $\phi_1 =0$. Then, we have
\begin{align}
    \alpha_1 &= r_1, \quad \alpha_3 = \pm ir_3, \\
    \alpha_2 &= \pm ir_2, \quad \alpha_4 = \pm r_4. 
\end{align}
Therefore,
\begin{align}
   &p_{\pm} = \frac{1}{2}|\alpha_1 \pm \alpha_4|^2 = \frac{1}{2}\left( r_1^2 + r_4^2 \pm 2r_1r_4\right), \\
   &q_{\pm} = \frac{1}{2}|\alpha_2 \pm \alpha_3|^2 = \frac{1}{2}\left( r_2^2 + r_3^2 \pm 2r_2r_3\right).
\end{align}
Consider the additional constraint (from $I_{\gamma_{xy}}$ =1) that
\begin{equation}\label{eq:constraint_1}
    -\sum_l p_l\log_2(p_l) = 1\,,
\end{equation}
where $\{p_l\}_{l} = \{p_+,p_-,q_+,q_-\}$. From \eqref{eq:constraints_amplitudes} we have that $r_2 = r_3$ and $r_1 = r_4$. Thus, $p_- = q_- =0$. Hence, the only solution of \eqref{eq:constraint_1} is $p_+ = q_+ = \frac{1}{2}$, which implies $r_i = \frac{1}{2}$ for all $i$.
\end{proof}

\section{Technical proofs}
\label{sec:technical details}
In this section, we prove Eq.~\eqref{eq:decomposition_to_prove_1}. First, let us consider a block decomposition for the Pauli string $P$ 
\begin{equation}\label{eq:block_decomposition}
    P=\begin{pmatrix}
        P\big|_{\rm GS} & B \\
        C & D
    \end{pmatrix}\,.
\end{equation}
Here we denoted by $P\big|_{\rm GS}$ the restriction of $P$ to the ground space. Given a ground state $\ket{\psi}$, we have
\begin{equation}
    P\ket{\psi}= P\big|_{\rm GS}\ket{\psi}+C\ket{\psi}\,.
\end{equation}
By construction, $C\ket{\psi}$ is a sum of excited states. Next, we claim
\begin{equation}\label{eq:commutation_restricted}
W_a(\gamma_y)    P\big|_{\rm GS}\ket{\psi}=P\big|_{\rm GS} W_a(\gamma_y)\ket{\psi}\,,
\end{equation}
for all $a\in \mathcal{A}$. To see this, we note that, since $P$ lies in the region $A_{\gamma_y}$, we can find a curve $\gamma_y'$ that is topologically equivalent to $\gamma_y$ and sufficiently distant, such that
\begin{equation}\label{eq:intermediate_eq}
P W_{a}(\gamma_y)\ket{\psi}=P W_{a}(\gamma_y')\ket{\psi}=W_{a}(\gamma_y')P\ket{\psi}\,,
\end{equation}
because the supports of $W_{a}(\gamma_y')$ and $P$ are disjoints. Here we used that the action of $W_{a}(\gamma_y)$ and $W_{a}(\gamma_y')$ coincide on the ground state. The l.h.s. of Eq.~\eqref{eq:intermediate_eq} can be rewritten as
\begin{align}\label{eq:lhs}
    P W_{a}(\gamma_y)\ket{\psi}&=P\big|_{\rm GS}W_{a}(\gamma_y)\ket{\psi}+CW_{a}(\gamma_y)\ket{\psi}\nonumber\\
    &=P\big|_{\rm GS}W_{a}(\gamma_y)\ket{\psi}+\ket{v_{\rm ex}}\,,
\end{align}
where $\ket{v_{\rm ex}}$ is a sum of excited states. The r.h.s. of Eq.~\eqref{eq:intermediate_eq} can be rewritten as
\begin{equation}
    W_{a}(\gamma_y')P\ket{\psi}=    W_{a}(\gamma_y')P\big|_{\rm GS}\ket{\psi}+ W_{a}(\gamma_y')C\ket{\psi}\,.
\end{equation}
In the string-net models, $W_a{(\gamma)}$ commutes with the Hamiltonian for any closed loop $\gamma$ (note that this is not true in general~\cite{cian2022extracting}), so that $W_a{(\gamma'_y)}$ is block diagonal with respect to the block decomposition~\eqref{eq:block_decomposition}. Therefore, $W_{a}(\gamma_y')C\ket{\psi}=\ket{w}_{\rm ex}$ is a sum of excited states, and we have
\begin{align}
\label{eq:rhs}
    W_{a}(\gamma_y')P\ket{\psi}&=    W_{a}(\gamma_y')P\big|_{\rm GS}\ket{\psi}+\ket{w}_{\rm ex}\nonumber\\
    &=  W_{a}(\gamma_y)P\big|_{\rm GS}\ket{\psi}+\ket{w}_{\rm ex}.
\end{align}
Equating~\eqref{eq:lhs} with~\eqref{eq:rhs}, and projecting onto the ground space we obtain~\eqref{eq:commutation_restricted}.

Eq.~\eqref{eq:commutation_restricted} tells us that $P\big|_{\rm GS}$ commutes with $W_{a}(\gamma)$ for all $a\in \mathcal{A}$. Because each common eigenstate of $\{W_{a}(\gamma)\}$ is uniquely identified by the corresponding set of eigenvalues, and since products of the form $W_{a}(\gamma)W_{b}(\gamma)$ can be written as a linear combination of the WLO $W_{a}(\gamma)$, it is easy to show that
\begin{equation}
P\big|_{\rm GS} =\sum_{d}\alpha_d W_d(\gamma_y)\big|_{\rm GS}\,,
\end{equation}
which proves Eq.~\eqref{eq:decomposition_to_prove_1}.

\section{Mutual information and lattice modular transformation}
\label{sec:proof_mutual_info_change_stringnets}

In this section, we prove Theorem~\ref{thrm:classification_string_nets}. To this end, we refine the argument lay out in Sec.~\ref{sec:full_classification_string_nets}, starting from Eq.~\eqref{eq:correct_implementation}. 
As in the main text, we consider $N_a$ additional ancillas. Let us fix the regions $A$ and $B$ containing only physical qubits, and denote by $C$ the complement of $A\cup B$ containing only physical qubits. For each region $R$ denote with $R^a$ the union of $R$ and all the ancillas which are closest to $R$. (Note that in the construction~\cite{zhu2020instantaneous} the ancillas are added homogeneously to the faces of the honeycomb lattice)

It is sufficient to prove the theorem for a Dehn twists such as in Fig.~\ref{fig:dehn_twist_lattice}, as Dehn twist along different axis generate all modular transformations. Thus, in the following one can consider for concreteness a permutation $\mathcal{P}_\zeta$ of the physical and ancilla qubits, similar to the shear transformation in Fig.~\ref{fig:dehn_twist_lattice}.

Let us defines the permutations of the regions $R^a$ as
\begin{align}
    \mathcal{P}^{-1}_\zeta: A^a\to \tilde{A}^a\,,\\
        \mathcal{P}^{-1}_\zeta: B^a\to \tilde{B}^a\,,\\
            \mathcal{P}^{-1}_\zeta: C^a\to \tilde{C}^a\,.
\end{align}
Note also that $\tilde C^a$ is obviously the complement of $\tilde A^a\cup \tilde B^a$ and because $P_\zeta$ is a shear transformation, the distance between $\tilde{A}^a$ and $\tilde{B}^a$ is of the same order as that of $A^a$ and $B^a$.

Analogously to the main text, we have
\begin{align}
\rho_{\tilde A^a\tilde B^a}&={\rm Tr}_{\tilde {C}^a}\left[\ket{\psi}\bra{\psi}\otimes (|0\rangle\langle 0|)^{\otimes N_a}\right]\nonumber\\
&= V\left({\rm Tr}_{C^a}[\mathcal{U}^\dagger_\zeta \ket{\varphi_\zeta}\bra{\varphi_\zeta}\mathcal{U}_{\zeta}]\right)V^\dagger\,.
\end{align}
Here, $V =V_{\tilde{A}^a\leftarrow A^a}\otimes V_{\tilde{B}^a\leftarrow B^a}$ is an isometry \be V:\HH_{A^a}\otimes \HH_{B^a} \to \HH_{\tilde{A}^a}\otimes \HH_{\tilde{B}^a}\ee with $\HH_{R^a}$ denoting the Hilbert space on the region $R^a$, while we used Eq.~\eqref{eq:correct_implementation} and defined \be \ket{\varphi_\zeta}=\mathcal{U}_\zeta \mathcal{P}_\zeta \left(|\psi\rangle\otimes \ket{0}^{N_a}\right) = \left(M_\zeta \ket{\psi}\right)\otimes \ket{0}^{N_a}.\ee

Therefore, 
\begin{equation}
    I_{\tilde{A}^a, \tilde{B}^a}\left(\ket{\psi}\otimes \ket{0}^{N_a}\right)= I_{A^a,B^a}\left(\mathcal{U}^\dagger_\zeta\ket{\varphi_{\zeta}}\right)\,.
\end{equation}
On the l.h.s. the ancillas are in a product state so that all entropies factorize and we can restrict to the physical qubits:
\be
    I_{\tilde{A}^a, \tilde{B}^a}\left(\ket{\psi} \otimes \ket{0}^{N_a}\right)= I_{\tilde{A}, \tilde{B}}\left(\ket{\psi}\right).
\ee
The r.h.s. is the mutual information of the state obtained by applying a shallow-depth quantum circuit to $\ket{\varphi_{\zeta}} $. It is immediate from the proof of Lemma~\ref{lemma:mut_info_invariance} that, because $\ket{\varphi_{\zeta}}$ is a ground state and $A^a$, $B^a$ are sufficiently separated regions, the mutual information of $\ket{\varphi_{\zeta}}$ is invariant under shallow-depth quantum circuits. Therefore,
\begin{align}
    I_{\tilde{A}, \tilde{B}}(\ket{\psi})&= I_{A^a,B^a}(\mathcal{U}^\dagger_\zeta\ket{\varphi_{\zeta}})\nonumber\\
    &=I_{A^a,B^a}(\ket{\varphi_{\zeta}})=I_{A,B}(M_\zeta \ket{\psi})\,,
\end{align}
where the last step is as for the l.h.s. This concludes the proof. 

\bibliographystyle{jsty3-author}
\bibliography{bib}

\end{document}